\documentclass[cleveref, autoref, thm-restate]{lipics-v2021}
\usepackage[sort]{cite}
\usepackage{tikz}
\usepackage{"mydefs"}
\usepackage{multirow}
\usepackage{apxproof}
\usepackage{colortbl}
\usepackage{thmtools}
\usepackage[justification=centering]{caption}
\usepackage{footnote}
\usepackage{anyfontsize}
\makesavenoteenv{tabular}
\usetikzlibrary{decorations.pathreplacing}
\nolinenumbers

\newtheorem{innercustomthm}{Theorem}
\newenvironment{customthm}[1]
  {\renewcommand\theinnercustomthm{#1}\innercustomthm}
  {\endinnercustomthm}

\Copyright{Marco Carmosino, Ronald Fagin, Neil Immerman, Phokion Kolaitis, Jonathan Lenchner, and Rik Sengupta}


\authorrunning{M. Carmosino, R. Fagin, N. Immerman, P. Kolaitis, J. Lenchner, and R. Sengupta}

\relatedversion{A shortened version of this paper is due to appear in the Proceedings of the 49th International Symposium on Mathematical Foundations of Computer Science (MFCS), 2024.}


\EventEditors{Rastislav Kr\'{a}lovi\v{c} and Anton\'{i}n Ku\v{c}era}
\EventNoEds{2}
\EventLongTitle{49th International Symposium on Mathematical Foundations of Computer Science (MFCS 2024)}
\EventShortTitle{MFCS 2024}
\EventAcronym{MFCS}
\EventYear{2024}
\EventDate{August 26--30, 2024}
\EventLocation{Bratislava, Slovakia}
\EventLogo{}
\SeriesVolume{306}
\ArticleNo{26}

\title{On the Number of Quantifiers Needed to Define Boolean Functions}

\author{Marco Carmosino}{IBM Research Cambridge, MA, USA}{mlc@ibm.com}{https://orcid.org/0009-0007-1118-1352}{}
\author{Ronald Fagin}{IBM Research Almaden, CA, USA}{fagin@us.ibm.com}{https://orcid.org/
0000-0002-7374-0347}{}
\author{Neil Immerman}{University of Massachusetts, Amherst, USA}{immerman@umass.edu}{https://orcid.org/0000-0001-6609-5952}{}
\author{Phokion G.\ Kolaitis}{UC Santa Cruz and IBM Research Almaden, CA, USA}{kolaitis@ucsc.edu}{https://orcid.org/0000-0002-8407-8563}{}
\author{Jonathan Lenchner\footnote{ Corresponding author.}}{IBM T.J. Watson Research Center, NY, USA}{lenchner@us.ibm.com}{https://orcid.org/0000-0002-9427-8470}{}
\author{Rik Sengupta}{IBM Research Cambridge, MA, USA}{rsengupta@cs.umass.edu}{https://orcid.org/0000-0002-9238-5408}{NSF CCF-1934846}
\date{}

\begin{document}

\maketitle

\begin{abstract}
The number of quantifiers needed to express first-order (FO) properties is captured by two-player combinatorial games called \emph{multi-structural} games. We analyze these games on binary strings with an ordering relation, using a technique we call \emph{parallel play}, which significantly reduces the number of quantifiers needed in many cases. 
Ordered structures such as strings have historically been notoriously difficult to analyze in the context of these and similar games. Nevertheless, in this paper, we provide essentially tight bounds on the number of quantifiers needed to characterize different-sized subsets of strings. 
The results immediately give bounds on the number of quantifiers necessary to define several different classes of Boolean functions. One of our results is analogous to Lupanov's upper bounds on circuit size and formula size in propositional logic: we show that every Boolean function on $n$-bit inputs can be defined by a FO sentence having $(1+\varepsilon)\frac{n}{\log(n)} + O(1)$ quantifiers, and that this is essentially tight. We reduce this number to $(1 + \varepsilon)\log(n) + O(1)$ when the Boolean function in question is sparse.
\end{abstract}

\maketitle

\section{Introduction}\label{sec:intro}

In 1981, Immerman \cite{MScanon0} introduced \emph{quantifier number} (QN) as a measure of the complexity of first-order (FO) sentences. For a function $g{:}~\mathbb{N} \to \mathbb{N}$, he defined $\textrm{QN}[g(n)]$ as the class of properties on $n$-element structures describable by a uniform sequence of FO sentences with $O(g(n))$ quantifiers. He then showed that on \emph{ordered} structures, for $f(n) \geq \log n$, one has:
\begin{equation} \label{neils-inlcusions}
  \textrm{NSPACE}[f(n)] \subseteq \textrm{QN}[(f(n))^2/\log n]\subseteq \textrm{DSPACE}[(f(n))^2],  
\end{equation}
thereby establishing an important connection between QN and space complexity and so directly linking a logical object to classical complexity classes. 

The same paper \cite{MScanon0} described a two-player combinatorial game (which Immerman called the \emph{separability game}), that captures quantifier number in the same way that the more well-known Ehrenfeucht-Fra\"{i}ss\'{e} (EF) game \cite{Ehr61, Fra54} captures quantifier rank (QR). The paper additionally showed that any property whatsoever of $n$-element \emph{ordered} structures can be described with a sentence having a QR of $\log n + 3$. Since a QR of $\log n + 1$ is required just to distinguish a linear order of size $n$ from smaller linear orders \cite{ROSENSTEIN:1982}, QR has limited power to distinguish properties over ordered structures. 
QN is potentially a more fine-grained and powerful measure for this purpose. However, owing to the inherent difficulties of the analysis of Immerman's separability game, the study of the game and of QN in general lay dormant for forty years, until the game was rediscovered and renamed the \emph{multi-structural} (MS) game in \cite{MScanon1}. In that paper the authors made initial inroads into understanding how to analyze the game, leading to several follow-up works~\cite{MScanon2,MScanon3,vinallsmeeth2024quantifier}. Other related games to study the number of quantifiers were recently introduced in \cite{hella2024}, and close cousins of MS games were used to study formula size in \cite{DBLP:journals/corr/abs-1208-4803,DBLP:journals/lmcs/GroheS05}. In \cite{DBLP:journals/corr/abs-1208-4803} the authors study a related problem to ours --- they examine the (existential) sentences of minimum size needed to express a particular set of string properties. However, even without the existential restriction, the connection between the minimum size of a sentence and its minimum number of quantifiers is not obvious. It is possible for a property to be expressible only by a much longer sentence with fewer quantifiers than one with more quantifiers.


The MS game is played by two players, Spoiler ($\bS$, he/him) and Duplicator ($\bD$, she/her), on two \emph{sets} $\cA, \cB$ of structures. Essentially, $\bS$ tries to break all partial isomorphisms between all pairs of structures (one from $\cA$ and the other from $\cB$) over a prescribed number of rounds, whereas $\bD$ tries to maintain a partial isomorphism between \emph{some} pair of structures. Unlike in EF games, $\bD$ has more power in MS games, since she can make arbitrarily many copies of structures before her moves, enabling her to play all possible responses to $\bS$'s moves. The fundamental theorem for MS games~\cite{MScanon0, MScanon1} (see Theorem~\ref{thm:MSfundamental}) states that $\bS$ has a winning strategy for the $r$-round MS game on $(\cA, \cB)$ if and only if there is a FO sentence $\varphi$ with at most $r$ quantifiers that is true for every structure in $\cA$ but false for every structure in $\cB$. We call such a $\varphi$ a \emph{separating sentence} for $(\cA, \cB)$. In general, our eventual objective will be to separate a set $\cA$ of $n$-bit strings from all other $n$-bit strings (i.e., from its complement $\cA^\mathrm{C}$). This is a particularly interesting question because of its intimate connection to the complexity of \emph{Boolean functions}.

\medskip \noindent{\bf Boolean Functions.} 
Any Boolean function on $n$-bit strings is specified by two complementary sets, $\cA, \cA^\textrm{C} \subseteq \{0, 1\}^n$, representing the input strings that get mapped to $1$ and $0$ respectively. 
For such a function $f \colon \{0,1\}^n \to \{0,1\}$, we say that a FO sentence $\phi$ in the vocabulary of strings \emph{defines} the function $f$ if $\phi$ is a separating sentence for $(f^{-1}(1), f^{-1}(0))$. Hence, the key results of this paper can be thought of as giving sharp bounds on the number of quantifiers needed to define Boolean functions. Our main results about the definability of Boolean functions are Theorems \ref{thm-A} and \ref{thm-B} below.

\begin{customthm} A Given an arbitrary $\varepsilon > 0$, every Boolean function on $n$-bit strings can be defined by a FO sentence having $(1 + \varepsilon)\frac{n}{\log(n)} + O_\varepsilon(1)$ quantifiers, where the $O_\varepsilon(1)$ additive term depends only on $\varepsilon$ and not $n$. Moreover, there are Boolean functions on $n$-bit strings that require $\frac{n}{\log(n)} + O(1)$ quantifiers to define. \label{thm-A}
\end{customthm}


Say that a family, 
$\{f_n\}_{n=1}^\infty$, of Boolean functions on $n$-bit strings, is \emph{sparse} if the cardinality of the set of strings mapping to $1$ under each $f_n$ is polynomial in $n$. For example, if $\mathcal{L}$ is a sparse language, then the family of Boolean functions, defined for each $n$, by the characteristic function of $\mathcal{L}$ restricted to $n$-bit inputs, is sparse \cite{sparsefortune, sparsemahaney}.

\begin{customthm} B Given an arbitrary $\varepsilon > 0$, and a sparse family, $\{f_n\}_{n=1}^\infty$, of Boolean functions on $n$-bit strings, each function $f_n$ can be defined by a FO sentence having $(1 + \varepsilon)\log(n) + O_\varepsilon(1)$ quantifiers, where the $O_\varepsilon(1)$ additive term depends only on $\varepsilon$ and not $n$ or the choice of sparse family. Moreover, there are sparse families of Boolean functions on $n$-bit strings, the functions of which require $\log(n)$ quantifiers to define. \label{thm-B}
\end{customthm}

Theorem \ref{thm-A} follows from Theorems \ref{thm:all-vs-all upper bound} and \ref{prop:all-vs-all lower bound} (in Section \ref{sec:strings}), whereas Theorem \ref{thm-B} follows from Theorem \ref{thm:one-vs-all-n} and Proposition \ref{prop:one-vs-one-n} (in Section \ref{sec:strings}). Theorem \ref{thm:all-vs-all upper bound} can be viewed as a first-order logic analog of the upper bounds obtained by Lupanov for minimum circuit size \cite{Lupanov1958} and minimum propositional formula size \cite{Lupanov1965} to capture an arbitrary Boolean function. Note that \emph{any} property whatsoever of $n$-bit strings can be captured trivially by a sentence with $n$ existential quantifiers. Similar to Lupanov's bounds, our result shows that we can shave off a factor of $\log(n)$ from this trivial upper bound.
Furthermore, Theorem \ref{prop:all-vs-all lower bound} establishes via a counting argument that there are functions with a QN lower bound that essentially matches our worst-case upper bound -- a result that can be viewed as a first-order logic analog of the Riordan-Shannon lower bound  \cite{riordanlower} for propositional formula size.
\medskip \noindent{\bf Parallel Play. }A key technical contribution we make in this paper is the Spoiler strategy of \emph{parallel play}, which widens the scope of winning strategies for $\bS$ compared to previous work. The essential idea is for $\bS$ to partition the sets $\cA$ and $\cB$ into subsets $\cA_1\sqcup\ldots\sqcup\cA_k$ and $\cB_1\sqcup\ldots\sqcup\cB_k$, and then play $k$ MS ``sub-games'' in parallel on $(\cA_i, \cB_i)$. In certain circumstances, $\bS$ can then combine his strategies for each of those sub-games into a strategy for the entire game, and thereby save many superfluous moves. Applying the fundamental theorem, this results in a very small number of quantifiers in the corresponding separating sentence.
\medskip \noindent{\bf Outline of the Paper. }
This paper is organized as follows. In Section \ref{sec:prelims}, we set up some preliminaries. In Section \ref{sec:parallel}, we precisely formulate what we call the \emph{Parallel Play Lemma} (Lemma \ref{lem:parallelplay}) and the \emph{Generalized Parallel Play Lemma} (Lemma \ref{lem:genparallelplay}). In Section \ref{sec:linearorders}, we develop results on linear orders that are similar to but more nuanced than those in \cite{MScanon1, MScanon2}, with the extra nuance being critical for our subsequent string separation results. 
In Section \ref{sec:strings}, we present our results on separating disjoint sets of strings. 
In Section \ref{sec:conclusion}, we wrap up with some conclusions and open problems.

\section{Preliminaries}\label{sec:prelims}

Fix a vocabulary $\tau$ with finitely many relation and constant symbols. We typically designate structures in boldface ($\bA$), their universes in capital letters ($A$), and sets of structures in calligraphic typeface ($\cA$). This last convention includes sets of pebbled structures (see below).

We always use $\log(\cdot)$ to designate the base-$2$ logarithm. Furthermore, in several results in Section \ref{sec:strings}, we have an $O(1)$ additive term. This term will always be independent of $n$. Any additional dependence will be stated in the form of a subscript on the $O$, e.g., $O_t(1)$ would denote a term independent of $n$, but dependent on the choice of some parameter $t$.

\medskip \noindent{\bf Pebbled Structures and Matching Pairs. } Consider a palette $\cC = \{\r, \b, \g, \ldots\}$ of \emph{pebble colors}, with infinitely many pebbles of each color available. A $\tau$-structure $\bA$ is \emph{pebbled} if some of its elements $a_1, a_2, \ldots \in A$ have pebbles on them. There can be at most one pebble of each color on a pebbled structure. There can be multiple pebbles (of different colors) on the same element $a_i \in A$. Occasionally, when the context is clear, we will use the term \emph{board} synonymously with ``pebbled structure''.
  
If $\mathbf{A}$ is a $\tau$-structure, and the first few pebbles are placed on elements $a_1, a_2, a_3 \ldots \in A$, we designate the resulting pebbled $\tau$-structure as $\langle \mathbf{A} ~|~ a_1, a_2, a_3, \ldots \rangle$. Note that $\bA$ can be viewed as a pebbled structure $\langle \bA ~|~\rangle$ with the empty set of pebbles.

By convention, we use $\r$, $\b$, and $\g$ for the first three pebbles we play (in that order), as a visual aid in our proofs. Hence, the pebbled structure $\langle \bA ~|~ a_1, a_2, a_3\rangle$ has pebbles $\r$ on $a_1 \in A$, $\b$ on $a_2 \in A$, and $\g$ on $a_3 \in A$. Note that $a_1$, $a_2$, and $a_3$ need not be distinct.


We say that the pebbled structures $\langle \bA ~|~ a_1, \ldots, a_k\rangle$ and $\langle \bB ~|~ b_1, \ldots, b_k\rangle$ are a \emph{matching pair} if the map $f \colon A \to B$ defined by:
\begin{itemize}
    \item $f(a_i) = b_i$ for all $1 \leq i \leq k$
    \item $f(c^\bA) = c^\bB$ for all constants $c$ in $\tau$
\end{itemize}
is an isomorphism on the induced substructures. Note that $\langle \bA ~|~ a_1, \ldots, a_k\rangle$ and $\langle \bB ~|~ b_1, \ldots, b_k\rangle$ can form a matching pair even when $\bA \not\cong \bB$.


\medskip \noindent{\bf Multi-Structural Games. } Assume $r \in \mathbb{N}$, and let $\cA$ and $\cB$ be two sets of pebbled structures, each pebbled with the \emph{same} set $\{x_1, \ldots, x_k\} \subseteq \cC$ of pebble colors. The \emph{$r$-round multi-structural (MS) game on $(\cA, \cB)$} is defined as the following two-player game, played by two players, \textbf{Spoiler} ($\bS$, he/him) and \textbf{Duplicator} ($\bD$, she/her). In each round $i$ for $1 \leq i \leq r$, $\bS$ chooses either $\cA$ or $\cB$, and an \textbf{unused} color $y_i \in \cC$; he then places (``plays'') a pebble of color $y_i$ on an element of \emph{every} board in the chosen set (``side''). In response, $\bD$ makes as many copies as she wants of each board on the other side, and plays a pebble of color $y_i$ on an element of each of those boards. $\bD$ wins the game if at the end of round $r$, there is a board in $\cA$ and a board in $\cB$ forming a matching pair. Otherwise, $\bS$ wins.
For readability, we always call the two sets $\cA$ and $\cB$, even though the structures change over the course of a game in two ways:
\begin{itemize}
    \item $\cA$ or $\cB$ can increase in size over the $r$ rounds, as $\bD$ can make copies of the boards.
    \item The number of pebbles on each of the boards in $\cA$ and $\cB$ increases by $1$ in each round.
\end{itemize}

We usually refer to $\cA$ as the \emph{left} side, and $\cB$ as the \emph{right} side.

Let $\cA$ and $\cB$ be two sets of pebbled structures, with each pebbled structure containing pebbles colored with $\{x_1, \ldots, x_k\} \subseteq \cC$. Let $\varphi(x_1, \ldots, x_k)$ be a FO formula with free variables $\{x_1, \ldots, x_k\}$. We say $\varphi$ is a \emph{separating formula} for $(\cA, \cB)$ (or $\varphi$ \emph{separates} $\cA$ and $\cB$) if:
\begin{itemize}
    \item for every $\langle \bA ~|~ a_1, \ldots, a_k\rangle \in \cA$ we have $\bA[a_1/x_1, \ldots, a_k/x_k] \models \varphi$,
    \item for every $\langle\bB ~|~ b_1, \ldots, b_k\rangle \in \cB$ we have $\bB[b_1/x_1, \ldots, b_k/x_k] \models \lnot\varphi$.
\end{itemize}
The following key theorem \cite{MScanon0, MScanon1}, stated here without proof, relates the logical characterization of a separating formula with the combinatorial property of a game strategy.

\begin{theorem}[Fundamental Theorem of MS Games, \cite{MScanon0, MScanon1}]\label{thm:MSfundamental}
    $\bS$ has a winning strategy in the $r$-round MS game on $(\cA, \cB)$ iff there is a formula with $\leq r$ quantifiers separating $\cA$ and $\cB$.
\end{theorem}

In the theorem above, if $\cA$ and $\cB$ are sets of \emph{unpebbled} structures, and $\varphi$ is a sentence, we call $\varphi$ a \emph{separating sentence} for $(\cA, \cB)$.

We note that $\bD$ has a clear optimal strategy in the MS game, called the \emph{oblivious} strategy: for each of $\bS$'s moves, $\bD$ can make enough copies of each pebbled structure on the other side to play all possible responses at the same time. If $\bD$ has a winning strategy, then the oblivious strategy is winning. For this reason, the MS game is essentially a single-player game, where $\bS$ can simulate $\bD$'s responses himself.

We make an easy observation here without proof, that will help us \emph{discard} some boards during gameplay; we can remove them without affecting the result of the game. This will help us in the analysis of several results in the paper.
\begin{observation}\label{obs:discard}
During gameplay in any instance of the MS game, consider a board $\langle \bA ~|~ a_1, \ldots, a_k\rangle$ such that there is no board on the other side forming a matching pair with it. Then, $\langle \bA ~|~ a_1, \ldots, a_k\rangle$ can be removed from the game without affecting the result.
\end{observation}

\medskip \noindent{\bf Linear Orders. } Let $\tau_\mathsf{ord} = \angle{< ;~ \mathsf{min},\mathsf{max}}$ be the vocabulary of orders, where $<$ is a binary predicate, and $\mathsf{min}$ and $\mathsf{max}$ are constant symbols. For every $\ell \geq 1$, we shall use $L_\ell$ to refer to a structure of type $\tau_\mathsf{ord}$, which interprets $<$ as a total linear order on $\ell + 1$ elements, and $\mathsf{min}$ and $\mathsf{max}$ as the first and last elements in that total order respectively. Note that there is only one linear order for any fixed value of $\ell$. When unambiguous, we may suppress the subscript and refer to the linear order as simply $L$.

We define the \emph{length} of a linear order $L$ as the size of its universe minus one (equivalently, as the number of edges if the linear order were represented as a path graph). Hence, the length of $L_\ell$ is $\ell$. Since we only consider $\ell \geq 1$, the length is always positive, and $\mathsf{min}$ and $\mathsf{max}$ are necessarily distinct. Our convention is different from \cite{MScanon1} and \cite{MScanon2}, where the length of a linear order was the number of elements, and the vocabulary had no built-in constants. Note that having $\mathsf{min}$ and $\mathsf{max}$ is purely for convenience; each can be defined and reused at the cost of two quantifiers.

Let $L$ be a linear order with elements $a < b$. The linear order $L[a, b]$ is the induced linear order on all elements from $a$ to $b$, both inclusive. If the variables $x$ and $y$ have been interpreted by $L$ so that $x^L = a$ and $y^L = b$, then we shall use $L[x, y]$ and $L[a, b]$ interchangeably; we adopt a similar convention for constants. If pebbles $\r$ and $\b$ have been placed on $L$ on $a$ and $b$ respectively, we use $L[\r, \b]$ to mean $L[a, b]$.

We will frequently need to consider sets of linear orders. For $\ell \geq 1$, we will use the notation $L_{\leq \ell}$ to denote the set of linear orders of length at most $\ell$, and $L_{> \ell}$ to denote the set of linear orders of length greater than $\ell$.

\medskip \noindent{\bf Strings. } Let $\tau_\mathsf{string} = \langle <, ~ S ~;~ \mathsf{min},\mathsf{max}\rangle$ be the vocabulary of binary strings, where $<$ is a binary predicate, $S$ is a unary predicate, and $\mathsf{min}$ and $\mathsf{max}$ are constant symbols. We encode a string $w = (w_1, \ldots, w_n) \in \{0,1\}^n$ by the $\tau_\mathsf{string}$-structure $\mathbf{B}_w$ having universe $B_w = \{1, \dots, n\}$, relation $<$ interpreted by the linear order on $\{1, \dots, n\}$, relation $S = \{ i ~|~ w_i = 1 \} $, and $\mathsf{min}$ and $\mathsf{max}$ interpreted as $1$ and $n$ respectively.

For an $n$-bit string $w$, and $i,j$ such that $1 \leq i \leq j \leq n$, denote by $w[i, j]$ the substring $w_i\ldots w_j$ of $w$. Note that $w[i, j]$ corresponds to the induced substructure of $\mathbf{B}_w$ on $\{i, \ldots, j\}$. We will often interchangeably talk about the string $w$ and the $\tau_\mathsf{string}$-structure $\mathbf{B}_w$, when the context is clear. As in $\tau_{\mathsf{ord}}$, having $\mathsf{min}$ and $\mathsf{max}$ in the vocabulary is purely for convenience.

\section{Parallel Play}\label{sec:parallel}

In this section, we prove our key lemma, that shows how, in certain cases, $\bS$ can combine his winning strategies in two sub-games, playing them in parallel in a single game that requires no more rounds than the longer of the two sub-games.

To understand why this is helpful, note that in general, if a formula $\varphi$ is of the form $\varphi_1\land\varphi_2$ or $\varphi_1\lor\varphi_2$, the number of quantifiers in $\varphi$ is the sum of the number of quantifiers in $\varphi_1$ and $\varphi_2$, even if the two subformulas have the same quantifier structure. We will see that playing parallel sub-games roughly corresponds to taking a $\varphi$ of the form $\varphi_1\land\varphi_2$ or $\varphi_1\lor\varphi_2$ where the subformulas have the same quantifier prefix, and writing $\varphi$ with the same quantifier prefix as $\varphi_1$ or $\varphi_2$, saving half the quantifiers we normally require.

Suppose $\bS$ has a winning strategy for an instance $(\cA, \cB)$ of the $r$-round MS game. In principle, the choice of which side $\bS$ plays on could depend on $\bD$'s previous responses. However, note that any strategy $\cS$ used by $\bS$ that wins against the oblivious strategy also wins against any other strategy that $\bD$ plays. Therefore, we may WLOG restrict ourselves to strategies used by $\bS$ against $\bD$'s oblivious strategy. It follows that the choice of which side to play on in every round is completely determined by the instance $(\cA, \cB)$, and independent of any of $\bD$'s responses. Let $\cS$ be such a winning strategy for $\bS$. We now define the \emph{pattern} of $\cS$, which specifies which side $\bS$ plays on in each round, when following $\cS$.

\begin{definition}\label{def:pattern}
Suppose $\cA$ and $\cB$ are sets of pebbled structures, and assume that $\bS$ has a winning strategy $\cS$ for the $r$-round $\ms$ game on $(\cA, \cB)$. The \emph{pattern} of $\cS$, denoted $\mathsf{pat}(\cS)$, is 
an $r$-tuple $(Q_1, \ldots, Q_r) \in \{\exists, \forall\}^r$, where:
\begin{equation*}
    Q_i = \begin{cases}
    \exists  & \text{ if $\bS$ plays in $\cA$ in round $i$,} \\
    \forall  & \text{ if $\bS$ plays in $\cB$ in round $i$.}
\end{cases}
\end{equation*}
We say that $\bS$ \emph{wins the game with pattern
$(Q_1,\ldots,Q_r)$} if $\bS$ has a winning strategy $\cS$ for the game in which $\mathsf{pat}(\cS) = (Q_1,\ldots,Q_r)$.
\end{definition}  

The following lemma is implicit in the proof of Theorem \ref{thm:MSfundamental}.

\begin{lemma}\label{lem:pattern}
For any two sets $\cA$ and $\cB$ of pebbled $\tau$-structures, the following are equivalent:
\begin{enumerate}
\item $\bS$ wins the $r$-round $\ms$ game on $(\cA, \cB)$ with pattern $(Q_1,\ldots,Q_r)$.
\item $(\cA,\cB)$ has a separating formula with $r$ quantifiers and quantifier prefix $(Q_1,\ldots,Q_r)$.
\end{enumerate}
\end{lemma}

Note that Lemma \ref{lem:pattern} implies that, as long as there is a separating formula $\varphi$ for $(\cA, \cB)$ with $r$ quantifiers, $\bS$ has a winning strategy for the $r$-round MS game on $(\cA, \cB)$ that ``follows'' $\varphi$; namely, if $\varphi = Q_1\ldots Q_r\psi$, then in round $i$, $\bS$ plays in $\cA$ if $Q_i = \exists$, and in $\cB$ if $Q_i = \forall$. Hence, for the rest of the paper, we will refer to $\bS$ moves in $\cA$ and $\cB$ as \emph{existential} and \emph{universal} moves respectively. We are now ready to state our main lemma from this section.

\begin{lemma}[Parallel Play Lemma]\label{lem:parallelplay}
Let  $\cA$ and $\cB$ be two sets of pebbled structures, and let $r \in \mathbb{N}$. Suppose that $\cA$ and $\cB$ can be partitioned as $\cA = \cA_1 \sqcup \cA_2$ and $\cB = \cB_1 \sqcup \cB_2$
respectively, such that for $1 \leq i \leq 2$, $\bS$ has a winning strategy $\cS_i$ for the $r$-round MS game on $(\cA_i,\cB_i)$, satisfying the following conditions:
\begin{enumerate}
\item Both $\cS_i$'s have the same pattern $P = \mathsf{pat}(\cS_1)=\mathsf{pat}(\cS_2)$.
\item At the end of the sub-games, both of the following are true:
\begin{itemize}
    \item There does not exist a board in $\cA_1$ and a board in $\cB_2$ forming a matching pair.
    \item There does not exist a board in $\cA_2$ and a board in $\cB_1$ forming a matching pair.
\end{itemize}
\end{enumerate}
Then $\bS$ wins the $r$-round MS game on $(\cA,\cB)$ with pattern $P$.
\end{lemma}

\begin{proof}
$\bS$ plays the $r$-round MS game on $(\cA,\cB)$ by playing his winning strategy $\cS_1$ on $(\cA_1,\cB_1)$, and his winning strategy $\cS_2$ on $(\cA_2, \cB_2)$, \emph{simultaneously} in parallel. This is a well-defined strategy, since every $\cS_i$ has the same pattern $P$. At the end of the game:
\begin{itemize}
\item for $i = j$, no board from $\cA_i$ forms a matching pair with a board from $\cB_j$, since $\bS$ wins the sub-game $(\cA_i, \cB_i)$.
\item for $i \neq j$, no board from $\cA_i$ forms a matching pair with a board from $\cB_j$, by assumption.
\end{itemize}
Therefore, no matching pair remains after round $r$, and so, $\bS$ wins the game. The pattern for this strategy is $P$ by construction.
\end{proof}

We observe that Lemma \ref{lem:parallelplay} can be generalized in two ways. Firstly, we could split into $k$ sub-games instead of two. Secondly, we can weaken assumption 1 in the statement of the lemma, so that each of the patterns is a subsequence of some $r$-tuple $P = \{\exists, \forall\}^r$. This is because $\bS$ can simply extend the strategy $\cS_i$ with pattern $P_i$ to a strategy $\cS'_i$ with pattern $P$, where for every ``missing'' entry in the tuple $P$, $\bS$ makes a dummy move on the corresponding side. We state a generalized version below without a proof.

\begin{lemma}[Generalized Parallel Play Lemma]\label{lem:genparallelplay}
Let  $\cA$ and $\cB$ be two sets of pebbled structures, and let $r \in \mathbb{N}$. Let $P \in \{\exists, \forall\}^r$ be a sequence of quantifiers of length $r$. Suppose that $\cA$ and $\cB$ can be partitioned as $\cA = \cA_1 \sqcup \ldots \sqcup \cA_k$ and $\cB = \cB_1 \sqcup \ldots \sqcup \cB_k$ respectively, such that for all $1 \leq i \leq k$, $\bS$ has a winning strategy $\cS_i$ for the $r_i$-round MS game on $(\cA_i,\cB_i)$ (where $r_i \leq r$), satisfying the following conditions:
\begin{enumerate}
\item For all $i$, $\mathsf{pat}(\cS_i)$ is a subsequence of $P$.
\item At the end of the sub-games, for $i \neq j$, there does not exist a board in $\cA_i$ and a board in $\cB_j$ forming a matching pair.
\end{enumerate}
Then $\bS$ wins the $r$-round MS game on $(\cA,\cB)$ with pattern $P$.
\end{lemma}

Note that Lemmas \ref{lem:parallelplay} and \ref{lem:genparallelplay} can be applied in conjunction with Observation \ref{obs:discard} as long as there is at least one structure remaining on either side, since a winning strategy (and therefore its corresponding pattern) is unaffected if some of the pebbled structures in the instance are deleted. Furthermore, in many cases, we can provide a strategy for $\bS$ where condition 2 in Lemmas \ref{lem:parallelplay} and Lemma \ref{lem:genparallelplay} will be automatically met after the first move, and therefore will continue to be satisfied at the end of the game. We shall use these two facts implicitly in the proofs that follow.
\section{Linear Orders}\label{sec:linearorders}


As noted in Section \ref{sec:intro}, the results in this section are similar to those in \cite{MScanon1, MScanon2}, but somewhat more nuanced, leading ultimately to the \emph{quantifier alternation theorems} (Theorems \ref{thm:alternation1} and \ref{thm:alternation2}). Instead of the unwieldy function $g(\cdot)$ studied in those papers, we study the simpler function $q(\cdot)$, which, given an integer $\ell$, returns the minimum number of quantifiers needed to separate $L_{\leq \ell}$ from $L_{>\ell}$. A key result, not appreciated in \cite{MScanon1, MScanon2}, is that the number of quantifiers needed to separate two linear orders of different sizes never exceeds the quantifier rank needed by more than one (Theorem \ref{thm:q_vs_r}).

\smallskip

Let $r(\ell)$ (resp.~$q(\ell)$) be the minimum QR (resp.~QN) needed to separate $L_{\leq \ell}$ and $L_{>\ell}$. Let $q_\forall(\ell)$ (resp.~$q_\exists(\ell)$) be the minimum number of quantifiers needed to separate $L_{\leq \ell}$ and $L_{>\ell}$ with a sentence whose prenex normal form starts with $\forall$ (resp.~$\exists$). Note that $q(\ell) = \min(q_\forall(\ell),q_\exists(\ell))$. The
values of $r(\ell)$ are well understood \cite{ROSENSTEIN:1982}:

\begin{theorem}[Quantifier Rank, \cite{ROSENSTEIN:1982}]\label{rosey-fact}
For $\ell\geq 1$, we have $r(\ell) = 1 + \floor{\log(\ell)}$.
\end{theorem}\quad

Since QR lower bounds QN, we have $r(\ell) \leq q(\ell)$ for all $\ell$. On the other hand, 
for each $\ell > 0$, we will show that $\bS$ can always separate $L_{\leq \ell}$ from $L_{>\ell}$ in a multi-structural game of at most $r(\ell) + 1$ rounds, which shows that $q(\ell) \leq r(\ell) + 1$.

For notational convenience, we denote by $\textrm{MSL}_{\exists,r}(\ell)$ an $r$-round MS game on $(L_{\leq \ell}, L_{> \ell})$, in which $\bS$ \emph{must} play an existential first round move. We use $\textrm{MSL}_{\forall,r}(\ell)$ analogously, where the first round move \emph{must} be universal. Observe that, \emph{a priori}, any such game may be winnable by either $\bS$ or $\bD$. Since we are primarily interested in upper bounds, we restrict our attention only to $\bS$-winnable games. We call such games simply \emph{winnable}.

\subsection{The Closest-to-Midpoint with Alternation Strategy}\label{sec:cma}

In this section, we describe a divide-and-conquer recursive strategy for $\bS$ to play winnable game instances $\textrm{MSL}_{Q,r}(\ell)$. 
This strategy will give us upper bounds on $q_\exists(\ell)$ and $q_\forall(\ell)$, which we will then relate to $r(\ell)$.

We define the \emph{closest-to-midpoint} of a linear order $L[x, y]$ as the element halfway between the elements corresponding to $x$ and $y$ if $L[x, y]$ has even length, or the element just left of center 
if $L[x, y]$ has odd length.

The $\bS$-winning strategy is called \emph{Closest-to-Midpoint with Alternation} ($\mathsf{CMA}$). The pattern for this strategy will alternate between $\exists$ and $\forall$, splitting each game recursively into two smaller sub-games that can be played in parallel using Lemma \ref{lem:parallelplay}. In these sub-games, placed pebbles will take on the roles of $\mathsf{min}$ and $\mathsf{max}$. $\bS$ continues in this way until the sub-games are on linear orders of length $2$ or less, at which point he can win them easily.

The idea is for $\bS$ to obey the following two rules throughout, except possibly the last three rounds:
\begin{itemize}
    \item $\bS$ starts on his designated side (determined by $Q$), and then alternates in every round;
    \item on every board, $\bS$ plays on the closest-to-midpoint of a linear order $L[x, y]$, chosen carefully to ensure he essentially ``halves'' the length of the instance every round.
\end{itemize}
Note that one consequence of the second point above is that $\bS$ will \textit{never} play on $\mathsf{max}$.

Before getting to a formal description of the strategy, let us illustrate the main idea through a worked example. Consider the (winnable) game $\textrm{MSL}_{\exists,4}(5)$. In round $1$, $\bS$ plays on the closest-to-midpoint of all boards in $L_{\leq 5}$ (by the two conditions in the $\mathsf{CMA}$ strategy). Before $\bD$'s response, we reach the position shown in Figure \ref{fig:strategy_sample_game}.

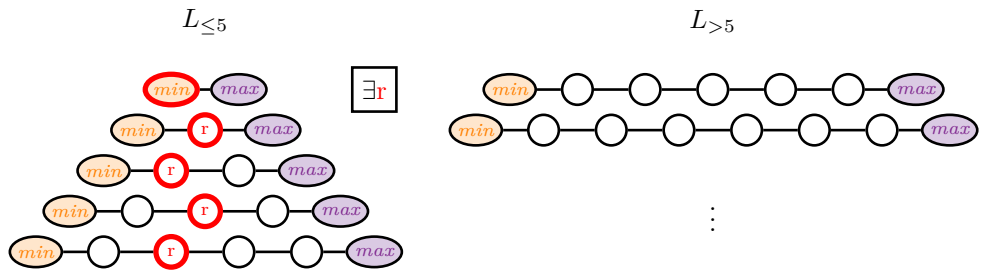
\begin{figure*}[ht]
\begin{center}
\tikzset{Red/.style={circle,draw=red, inner sep=0pt, minimum size=.9cc, line width=2pt}}
\tikzset{RedOuter/.style={circle,draw=red, inner sep=0pt, minimum size=1.2cc, line width=2pt}}
\tikzset{BlueOuter/.style={circle,draw=blue, inner sep=0pt, minimum size=1.25cc, line width=2pt}}
\tikzset{Blue/.style={circle,draw=blue, inner sep=0pt, minimum size=.9cc, line width=2pt}}
\tikzset{Green/.style={circle,draw=dg, inner sep=0pt, minimum size=.9cc, line width=2pt}}
\tikzset{Black/.style={circle,draw=black, inner sep=0pt, minimum size=.9cc, line width=1pt}}
\tikzset{TreeNode/.style={rectangle,draw=black, inner sep=0pt, minimum size=1.3cc, line width=1pt}}
\tikzset{BigTreeNode/.style={rectangle,draw=black, inner sep=0pt, minimum size=1.8cc, line width=1pt}}
\tikzset{Tleft/.style={ellipse,draw=black, inner sep=0pt, minimum size=.9cc,
    line width=1pt,fill=orange!20}}
\tikzset{TRleft/.style={ellipse,draw=red, inner sep=0pt, minimum size=.9cc,
    line width=2pt,fill=orange!20}}
\tikzset{Tright/.style={ellipse,draw=black, inner sep=0pt, minimum size=.9cc,
    line width=1pt,fill=indigo!20}}

\begin{tikzpicture}[scale=.09]
\node at (-35,74) {$L_{\leq 5}$};
\node at (40,74) {$L_{> 5}$};
\node at (40,46) {$\vdots$};

\node [TreeNode] (T0) at(-10,64)  {$\exists \r$};
\node [TRleft] (a0) at (-40,64) {{\scriptsize \mn}};
\node [Tright] (a1) at (-30,64) {{\scriptsize \mx}};
\node [Tleft] (b0) at (10,64) {{\scriptsize \mn}};
\node [Black] (b1) at (20,64) {{\scriptsize $$}};
\node [Black] (b2) at (30,64) {{\scriptsize $$}};
\node [Black] (b3) at (40,64) {{\scriptsize $$}};
\node [Black] (b4) at (50,64) {{\scriptsize $$}};
\node [Black] (b5) at (60,64) {{\scriptsize $$}};
\node [Tright] (b6) at (70,64) {{\scriptsize \mx}};

\foreach \from/\to in {a0/a1,b0/b1,b1/b2,b2/b3,b3/b4,b4/b5,b5/b6}
\draw[-,line width=1pt,color=black] (\from) -- (\to);

\node [Tleft] (1a0) at (-45,58) {{\scriptsize \mn}};
\node [Red]   (1a1) at (-35,58) {{\scriptsize $\r$}};
\node [Tright] (1a2) at (-25,58) {{\scriptsize \mx}};
\node [Tleft] (1b0) at (5,58) {{\scriptsize \mn}};
\node [Black] (1b1) at (15,58) {{\scriptsize $$}};
\node [Black] (1b2) at (25,58) {{\scriptsize $$}};
\node [Black] (1b3) at (35,58) {{\scriptsize $$}};
\node [Black] (1b4) at (45,58) {{\scriptsize $$}};
\node [Black] (1b5) at (55,58) {{\scriptsize $$}};
\node [Black] (1b6) at (65,58) {{\scriptsize $$}};
\node [Tright] (1b7) at (75,58) {{\scriptsize \mx}};
\foreach \from/\to in {1a0/1a1,1a1/1a2,1b0/1b1,1b1/1b2,1b2/1b3,1b3/1b4,1b4/1b5,1b5/1b6,1b6/1b7}
\draw[-,line width=1pt,color=black] (\from) -- (\to);

\node [Tleft] (2a0) at (-50,52) {{\scriptsize \mn}};
\node [Red] (2a1) at (-40,52) {{\scriptsize $\r$}};
\node [Black] (2a2) at (-30,52) {{\scriptsize $$}};
\node [Tright] (2a3) at (-20,52) {{\scriptsize \mx}};
\foreach \from/\to in {2a0/2a1,2a1/2a2,2a2/2a3}
\draw[-,line width=1pt,color=black] (\from) -- (\to);

\node [Tleft] (3a0) at (-55,46) {{\scriptsize \mn}};
\node [Black] (3a1) at (-45,46) {{\scriptsize $$}};
\node [Red] (3a2) at (-35,46) {{\scriptsize $\r$}};
\node [Black] (3a3) at (-25,46) {{\scriptsize $$}};
\node [Tright] (3a4) at (-15,46) {{\scriptsize \mx}};
\foreach \from/\to in {3a0/3a1,3a1/3a2,3a2/3a3,3a3/3a4}
\draw[-,line width=1pt,color=black] (\from) -- (\to);

\node [Tleft] (4a0) at (-60,40) {{\scriptsize \mn}};
\node [Black] (4a1) at (-50,40) {{\scriptsize $$}};
\node [Red] (4a2) at (-40,40) {{\scriptsize $\r$}};
\node [Black] (4a3) at (-30,40) {{\scriptsize $$}};
\node [Black] (4a4) at (-20,40) {{\scriptsize $$}};
\node [Tright] (4a5) at (-10,40) {{\scriptsize \mx}};
\foreach \from/\to in {4a0/4a1,4a1/4a2,4a2/4a3,4a3/4a4,4a4/4a5}
\draw[-,line width=1pt,color=black] (\from) -- (\to);
\end{tikzpicture}
\end{center}
\caption{The position after $\bS$'s round $1$ move in the game $\textrm{MSL}_{\exists,4}(5)$. The pebble $\r$ is on the closest-to-midpoint of every board on the left.}
\label{fig:strategy_sample_game}
\end{figure*}  

Now assume $\bD$ responds obliviously. We can first use Observation \ref{obs:discard} to discard all boards on the right with $\r$ on $\mathsf{max}$. By virtue of $\bS$'s first move, every board $\langle L ~|~ a_1\rangle$ on the left satisfies \emph{both} $L[\mathsf{min}, \r] \leq 2$, and $L[\r, \mathsf{max}] \leq 3$. Now consider any board $\langle L' ~|~ a'_1\rangle$ on the right. Note that either $L'[\mathsf{min}, \r] > 2$, or $L'[\r, \mathsf{max}] > 3$. Partition the right side as $\cB_1 \sqcup \cB_2$, where every $\langle L' ~|~ a'_1\rangle \in \cB_1$ satisfies $L'[\mathsf{min}, \r] > 2$, and every $\langle L' ~|~ a'_1\rangle \in \cB_2$ satisfies $L'[\r, \mathsf{max}] > 3$.

In round $2$, $\bS$ makes a universal move (by the first condition in the $\mathsf{CMA}$ strategy). In all boards in $\cB_1$, he plays pebble $\b$ on the closest-to-midpoint of $L'[\mathsf{min}, \r]$; similarly, in all boards in $\cB_2$, he plays pebble $\b$ on the closest-to-midpoint of $L'[\r, \mathsf{max}]$. Note that in either case, $\bS$ plays $\b$ on an element which is not on $\r$, $\mathsf{min}$, or $\mathsf{max}$.

After $\bD$ responds obliviously, we can use Observation \ref{obs:discard} to discard all boards on the left where $\b$ is on $\mathsf{min}$, $\mathsf{max}$, or $\r$. Since in particular this discards all boards on the left with $\r$ on $\mathsf{min}$, we can again use Observation \ref{obs:discard} to discard all boards from the right which have $\r$ on $\mathsf{min}$. Every remaining board in $\cB_1$ (resp.~$\cB_2$) corresponds to the isomorphism class $\mathsf{min} < \b < \r < \mathsf{max}$ (resp.~$\mathsf{min} < \r < \b < \mathsf{max}$). The remaining boards on the left also correspond to exactly one of those classes. Partition the left side as $\cA_1 \sqcup \cA_2$ accordingly.

Now, because of this difference in isomorphism classes, we will never obtain a matching pair from $\cA_1$ and $\cB_2$ (or from $\cA_2$ and $\cB_1$). Furthermore, for the rest of the game, $\bS$ will \emph{only} play inside $L[\mathsf{min}, \r]$ on all boards in $\cA_1$ and $\cB_1$, and inside $L[\r, \mathsf{max}]$ on all boards in $\cA_2$ and $\cB_2$. Suppose, in response to such a move on $\cA_1$, $\bD$ plays outside the range $L[\mathsf{min}, \r]$ on a board from $\cB_1$; the resulting board cannot form a partial match with any board from $\cA_1$ (since there is a discrepancy with $\r$), or with any board from $\cA_2$ (as observed already). Therefore, this board from $\cB_1$ can be discarded using Observation \ref{obs:discard}. A similar argument applies if $\bD$ ever responds outside the corresponding range in $\cB_2$, $\cA_1$, or $\cA_2$.

It follows that the sub-game $(\cA_1, \cB_1)$ (resp.~$(\cA_2, \cB_2$) corresponds \emph{exactly} to the game $\textrm{MSL}_{\forall, 3}(2)$ (resp.~$\textrm{MSL}_{\forall, 3}(3)$) where $\bS$ has already made his first move using the $\mathsf{CMA}$ strategy by playing a universal move on the closest-to-midpoints of the (relevant) linear orders. Since $\bS$ will alternate sides throughout, the patterns for both sub-game strategies will be the same.

We can now apply Lemma \ref{lem:parallelplay}. Observe that the lengths of the instances in the sub-games have been roughly halved, at the cost of a single move. The game then proceeds as shown in Figure \ref{fig:small_game_tree}. The leaves of the tree correspond to base cases (analyzed in Section \ref{sec:cmaformal}). The pattern of the strategy is preserved along all branches.
\begin{figure}[ht]
    \centering
    \includegraphics[scale=0.6]{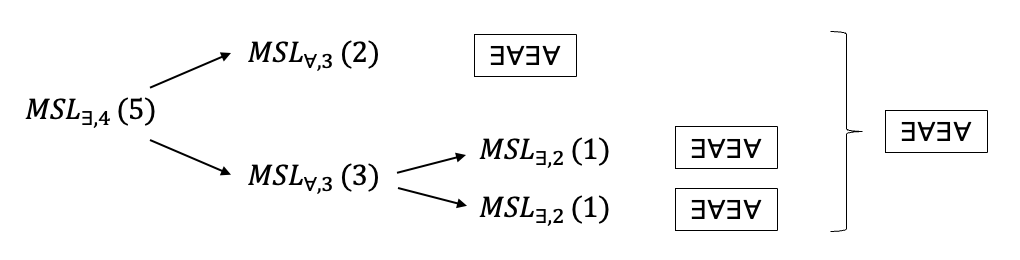}
    \caption{The $\textrm{MSL}_{\exists,4}(5)$ game tree. Each leaf is decorated with the associated quantifier prefix. All paths can be played in parallel using Lemma \ref{lem:genparallelplay} using the pattern $(\exists, \forall, \exists, \forall)$.}
    \label{fig:small_game_tree}
\end{figure}

\subsection{Formalizing the Strategy}\label{sec:cmaformal}

The first step in formalizing the $\mathsf{CMA}$ strategy for $\bS$ is to define four base cases, which we shall call \textit{irreducible} games. We assert the following (see Proposition \ref{prop:irreduciblepatterns} in Appendix \ref{app:lo-proofs}).
\begin{enumerate}
    \item $\textrm{MSL}_{\forall, 1}(1)$ is winnable with the pattern $(\forall)$.
    \item $\textrm{MSL}_{\exists, 2}(1)$ is winnable with the pattern $(\exists, \forall)$.
    \item $\textrm{MSL}_{\forall, 2}(2)$ is winnable with the pattern $(\forall, \forall)$.
    \item $\textrm{MSL}_{\forall, 3}(2)$ is winnable with the pattern $(\forall, \exists, \forall)$.
\end{enumerate}

The game $\textrm{MSL}_{\exists, 1}(1)$ is not winnable and hence not considered.

We now give a formalization of the inductive step. For a given quantifier $Q \in \{\exists, \forall\}$ and its complementary quantifier $\bar{Q}$, consider the game $\textrm{MSL}_{Q,k}(\ell)$. Note that if $\bS$ employs the $\mathsf{CMA}$ strategy 
the game splits into the two sub-games $\textrm{MSL}_{\bar{Q},k-1}(\ell')$ and $\textrm{MSL}_{\bar{Q},k-1}(\ell'')$. We designate this split as:
\begin{equation*}
    \textrm{MSL}_{Q,k}(\ell) \rightarrow  \textrm{MSL}_{\bar{Q},k-1}(\ell') \oplus \textrm{MSL}_{\bar{Q},k-1}(\ell'').
\end{equation*}
We will show in the proof of Lemma \ref{lem:cma-is-well-defined} that these sub-games can be played recursively, in parallel. When $\bS$ reaches an irreducible sub-game, he plays the winning patterns asserted above. We claim the following about the rules for splitting. The proof is in Appendix \ref{apx:linearorders}.

\begin{restatable}[Splitting Rules]{claim}{splitrules}\label{claim:splitrules}
    For $k \geq 3$, we have:
\begin{eqnarray} \label{msl-cases}
    \textrm{(i)~MSL}_{\exists,k}(2\ell) \rightarrow  \textrm{MSL}_{\forall,k-1}(\ell) \oplus \textrm{MSL}_{\forall,k-1}(\ell),~~~~~~~~~~~~~~~~~~~\ell \geq 1 \notag\\
    \textrm{(ii)~MSL}_{\exists,k}(2\ell + 1) \rightarrow  \textrm{MSL}_{\forall,k-1}(\ell) \oplus \textrm{MSL}_{\forall,k-1}(\ell+1),~~~~~~~~~\ell \geq 1 \\
    \textrm{(iii)~MSL}_{\forall,k}(2\ell) \rightarrow \textrm{MSL}_{\exists,k-1}(\ell) \oplus \textrm{MSL}_{\exists,k-1}(\ell-1),~~~~~~~~~~~~~~\ell \geq 2 \notag \\
    \textrm{(iv)~MSL}_{\forall,k}(2\ell + 1) \rightarrow  \textrm{MSL}_{\exists,k-1}(\ell) \oplus \textrm{MSL}_{\exists,k-1}(\ell),~~~~~~~~~~~~~~\ell \geq 1 \notag
\end{eqnarray}
\end{restatable}
Of course, the $\mathsf{CMA}$ strategy starts out seemingly promisingly, splitting with both initial sub-games starting on the same side; we must ensure that the strategy continues to be \emph{well-defined}, i.e., this continues throughout the recursion stack, especially since the sub-games can have different lengths. We show this in Lemma \ref{lem:cma-is-well-defined}, whose proof is in Appendix \ref{apx:linearorders}.

\begin{restatable}{lemma}{cmawelldefined}\label{lem:cma-is-well-defined}
The $\mathsf{CMA}$ strategy is well-specified. Moreover, for $k \geq 3$, if $\textrm{MSL}_{Q,k}(\ell) \rightarrow  \textrm{MSL}_{\bar{Q},k-1}(\ell_1) \oplus \textrm{MSL}_{\bar{Q},k-1}(\ell_2)$ with $\ell_1 \geq \ell_2$, then the pattern of $\bS$'s winning strategy for $\textrm{MSL}_{Q,k}(\ell)$ is $Q$ concatenated with the pattern for the winning strategy for $\textrm{MSL}_{\bar{Q},k-1}(\ell_1)$.
\end{restatable}

\subsection{Bounding and Characterizing the Pattern}\label{sec:boundingqn}

Define $q^*_\exists(\ell)$ (resp.~$q^*_\forall(\ell)$) as the minimum $r \in \mathbb{N}$ such that $\bS$ wins the game $\textrm{MSL}_{\exists,r}(\ell)$ (resp.~$\textrm{MSL}_{\forall,r}(\ell)$) using the $\mathsf{CMA}$ strategy. Of course, we must have $q_\exists(\ell) \leq q^*_\exists(\ell)$ and $q_\forall(\ell) \leq q^*_\forall(\ell)$. Let $q^*(\ell) = \min(q^*_\exists(\ell), q^*_\forall(\ell))$. The following lemma (whose proof is omitted) follows from the complete description of the strategy from Section \ref{sec:cmaformal}.

\begin{lemma}\label{lem:q*-vals}
We have $q^*_\forall(1) = 1$, $q^*_\exists(1) = 2$, and $q^*_\forall(2) = 2$. Also:
\begin{align*}
        &q^*_\exists(2\ell) = q^*_\forall(\ell) + 1~~~~\textrm{for }\ell \geq 1, 
        &&&q^*_\exists(2\ell+1) = q^*_\forall(\ell+1) + 1~~~~\textrm{for }\ell \geq 1, \\
        &q^*_\forall(2\ell) = q^*_\exists(\ell) + 1~~~~\textrm{for }\ell \geq 2,
        &&&q^*_\forall(2\ell+1) = q^*_\exists(\ell) + 1~~~~~~~~~\textrm{for }\ell \geq 1.
\end{align*}
\end{lemma}

From Lemma \ref{lem:q*-vals} it is possible to recursively compute $q^*_\forall(\ell)$ and $q^*_\exists(\ell)$, and therefore $q^*(\ell)$ for all values of $\ell \geq 1$. These values are provided for $\ell \leq 127$ in Table \ref{q*-table}.

\begin{table}[ht] 
\begin{center}
\begin{tabular}{|c|c|c|c|c|}\hline 
$\ell$ & $q^*_\forall(\ell)$ & $q^*_\exists(\ell)$ & $q^*(\ell)$ & $r(\ell)$ \\ \hline \hline
  1 & 1 & 2& 1 & 1\\ \hline
  2 & 2 & 2 & 2 & 2 \\ \hline
  3  & 3 & 3 & 3 & 2\\ \hline
  4 & 3 & 3& 3 & 3\\ \hline
  5 & 3 & 4& 3& 3\\ \hline
  6-7 & 4 & 4 & 4 & 3\\ \hline
  8-9 & 4 & 4 & 4 & 4\\ \hline
  10 & 5 & 4 & 4 & 4\\ \hline
  11-15 & 5 & 5 & 5 & 4\\ \hline
  16-18 & 5 & 5 & 5 & 5\\ \hline
  19-21 & 5 & 6 & 5 & 5\\ \hline
  22-31 & 6 & 6 & 6 & 5\\ \hline
  32-37 & 6 & 6 & 6 & 6\\ \hline
  38-42 & 7 & 6 & 6 & 6\\ \hline
  43-63 & 7 & 7 & 7 & 6\\ \hline
  64-75 & 7 & 7 & 7 & 7\\ \hline
  76-85 & 7 & 8 & 7 & 7\\ \hline
  86-127 & 8 & 8 & 8 & 7\\ \hline
\end{tabular}
\caption{Values of $q^*_\forall(\ell), q^*_\exists(\ell), q^*(\ell)$ and $r(\ell)$ for $1\leq \ell \leq 127.$} 
\label{q*-table}
\end{center}
\end{table}

We now state and prove the main result of this section.

\begin{restatable}{theorem}{mainlos} \label{thm:q_vs_r}
For all $\ell \geq 1$, we have:
\begin{equation*}
    r(\ell) \leq q(\ell) \leq r(\ell) + 1.
\end{equation*}
\end{restatable}

\begin{proof}
The first inequality, $r(\ell) \leq q(\ell)$, is obvious. For the second, we will show that $q^*_\exists(\ell)$ and $q^*_\forall(\ell)$ are both bounded above by $r(\ell) + 1$ (and since $q(\ell) \leq q^*(\ell) = \min(q^*_\exists(\ell), q^*_\forall(\ell))$, so too for $q(\ell)$). Lemma \ref{lem:q*-vals} shows that the assertion is true for $\ell \leq 2$. Now, it can be shown recursively (see, e.g., Proposition \ref{prop:q*-powers-of-two} in Appendix \ref{apx:linearorders}) that $q^*_\forall(2^k) = q^*_\exists(2^k) = k+1$ for $k \geq 1$. By Theorem \ref{rosey-fact}, we also know that $r(2^k) = k+1$ for $k \geq 1$. So the three functions, $r(\cdot)$, $q^*_\forall(\cdot)$, and $q^*_\exists(\cdot)$, all equal each other at successive powers of two, and increase by one between these successive powers. Since all three functions are monotonic, they differ from one another by at most one. Therefore, we have $q^*_\exists(\ell) \leq r(\ell) + 1$ and $q^*_\forall(\ell) \leq r(\ell) + 1$.
\end{proof}

We wrap up this section with two results that will be useful in Section \ref{sec:strings}. For their proofs, please see Appendix \ref{apx:linearorders}.


\begin{restatable}[Alternation Theorem, Smaller vs.~Larger]{theorem}{alternationone}\label{thm:alternation1} 
For every $\ell \geq 1$, there is a separating sentence $\sigma_\ell$ for $(L_{\leq \ell}, L_{> \ell})$ with $q^*(\ell)$ quantifiers (and so at most $\log(\ell) + 2$ quantifiers), such that the quantifier prefix of $\sigma_\ell$ strictly alternates and ends with a $\forall$. 
\end{restatable}

\begin{restatable}[Alternation Theorem, One vs.~All]{theorem}{alternationtwo}\label{thm:alternation2} 
For every $\ell \geq 1$, there is a sentence $\varphi_\ell$ separating $L_\ell$ from all other linear orders having an alternating quantifier prefix (ending with a $\forall$) and consisting of $q^*(\ell)+2$ quantifiers (and so at most $\log(\ell) + 4$ quantifiers).
\end{restatable}

\section{Strings}\label{sec:strings}

In this section, we pursue our main objective: string separation results, in order to characterize the complexity of Boolean functions. We would like to bound the number of quantifiers required for these separations as a function of both the length $n$ of the strings, as well as the sizes of the sets.

In general, we would like to separate a set of $n$-bit strings from the set of all other $n$-bit strings; recall from Section \ref{sec:intro} that we can think of this as separating the $1$ instances from the $0$ instances for a Boolean function on $n$-bit inputs. To do so, we first need to develop a basic technique for \emph{distinguishing} one string from another.

\begin{restatable}[One vs.~One]{proposition}{onevsonen}\label{prop:one-vs-one-n}
    \textbf{Upper Bound:} For every pair $w, w'$ of $n$-bit strings such that $w \neq w'$, there is a sentence $\varphi_{w, w'}$ with $\log(n) + 6$ quantifiers separating $(\{w\}, \{w'\})$. This sentence $\varphi_{w, w'}$ (in prenex form) has an alternating quantifier prefix ending with $\forall$.\\
    \textbf{Lower Bound:} For all sufficiently large $n$, there exist two $n$-bit strings $w, w'$, such that separating them requires $\lfloor\log(n)\rfloor$ quantifiers.
\end{restatable}

\begin{proof}
    \textbf{Upper Bound:} Let $w, w' \in \{0,1\}^n$ be any two distinct $n$-bit strings. There is an index $i \in [n]$ such that $w_i \neq w'_i$. Let $\cA = \{w\}$ and $\cB = \{w'\}$. We will show that $\bS$ wins the MS game on $(\cA, \cB)$ in $\log(n) + 6$ rounds.
    
    In round $1$, $\bS$ plays pebble $\r$ on the $\cA$ side, on the element $w_i$ in $w$, creating the pebbled string $\langle w ~|~ w_i\rangle$. Assume $\bD$ responds obliviously on the $\cB$ side. We can now immediately use Observation \ref{obs:discard} to discard the resulting pebbled string $\langle w' ~|~ w'_i \rangle \in \cB$, where the pebble $\r$ is on the element $w'_i$. Every remaining board in $\cB$ is of the form $\langle w' ~|~ w'_j\rangle$, for $j \neq i$. Note that the substring $w'[1, j]$ has length $j$, which is different from $i$, the length of the substring $w[1, i]$ of $w \in \cA$. So now, $\bS$ can simply play the strategy from Theorem \ref{thm:alternation2} to separate a linear order of length $i$ from all other linear orders, which he wins in $\log(n) + 4$ rounds with an alternating pattern. This gives us the desired result, after at most one more dummy move to preserve alternation.
    
    \textbf{Lower Bound:} Let $\ell = 2^k + 2$ for $k > 1$, and let $w = 0^{2^{k-1}} 100^{2^{k-1}}$ and $w' = 0^{2^{k-1}} 010^{2^{k-1}}$. 
    If $\bS$ plays entirely on one side of the respective $1$s then he is effectively playing the MS game on $(L_{2^{k-1}}, L_{2^{k-1}-1})$. By Theorem \ref{rosey-fact}, we have $r(2^{k-1}) = k = \lfloor \log(\ell) \rfloor$. Since QR lower bounds QN, the MS game played in this fashion requires at least $\lfloor \log(\ell)\rfloor$ rounds to win.

    Now suppose that instead of playing entirely on the same side of the respective $1$s, $\bS$ plays on both sides of a $1$ and/or on the $1$ during these $\lfloor \log(\ell)\rfloor$ rounds. In this case, $\bD$ can play obliviously to the left of the 1 when $\bS$ plays to the left of the 1, obliviously to the right of the 1 when $\bS$ plays to the right of the 1, and on the $1$ whenever $\bS$ plays on the $1$, thereby keeping matching pairs simultaneously on both sides. The lower bound follows.
\end{proof}

%




We also need another helpful lemma, whose proof is in Appendix \ref{apx:strings}.

\begin{restatable}{lemma}{stirling}\label{lem:stirling}
    Let $f \colon \mathbb{N} \to \mathbb{N}$ be a function satisfying $\lim_{n \to \infty}f(n) = \infty$, and let $t \geq 2$ be any integer. Then, for some number $N(t)$ depending on $t$, for all $n \geq N(t)$, we have $\lceil\log_t(f(n))\rceil! \geq f(n)$.
\end{restatable}

We now start with our string separation problems. The first problem we will consider will be when there is a single $n$-bit string in $\cA$, and the $2^n - 1$ remaining $n$-bit strings in $\cB$. Note that this corresponds to our Boolean function of interest being an indicator function.

\begin{restatable}[One vs.~All]{theorem}{onevsalln}\label{thm:one-vs-all-n}
For all $n$, and for every $\varepsilon > 0$, it is possible to separate each $n$-bit string from all other $n$-bit strings by a sentence with $(1 + \varepsilon)\log(n) + O_\varepsilon(1)$ quantifiers. This sentence (in prenex form) starts with a $\forall$, then has at most $\varepsilon\log(n) + 1$ occurrences of $\exists$, and then ends with an alternating quantifier prefix of length at most $\log(n) + O_\varepsilon(1)$.
\end{restatable}

\begin{proof}[Proof Sketch (see Appendix \ref{apx:strings} for full proof)]
Fix any $\varepsilon > 0$, and fix any integer $t \geq 2^{1/\varepsilon}$. By Lemma \ref{lem:stirling}, we know there is some integer $N(t)$, such that for all $n \geq N(t)$, $\lceil\log_t(n)\rceil! \geq n$. For any such $n$, fix an arbitrary $w \in \{0, 1\}^n$, and let $\cA = \{w\}$, and $\cB = \{0, 1\}^n - \{w\}$.

Consider the MS game on $(\cA, \cB)$. Every $w' \in \cB$ differs from $w$ in at least one bit. In round $1$, $\bS$ plays a universal move, placing a pebble on each $w' \in \cB$ on an index that disagrees with $w$ at that index. Assume $\bD$ responds obliviously, so that there are $n$ resulting pebbled strings in $\cA$. For the next $\lceil\log_t(n)\rceil$ rounds, $\bS$ plays only existential moves, placing the $\lceil\log_t(n)\rceil$ pebbles in distinct permutations on the $n$ strings in $\cA$, creating $n$ distinct isomorphism classes\footnote{An \emph{isomorphism class} is a maximal set of partially isomorphic pebbled structures.} by Lemma \ref{lem:stirling}. Once we discard structures from the two sides using Observation \ref{obs:discard}, we are now left with $n$ isomorphism classes, each of them defining a \textbf{one-vs-all} sub-game; in each of these sub-games, the round $1$ pebble is placed at a different index in the single string on the left from any string on the right. Therefore, $\bS$ can view this as a game simply about lengths, and can employ any \textbf{one-vs-all} linear order strategy. The entire game therefore reduces to $n$ parallel instances of \textbf{one-vs-all} sub-games on linear orders.

By Lemma \ref{lem:genparallelplay} and Theorem \ref{thm:alternation2}, $\bS$ can now win these parallel games in $\log(n) + 4$ further moves. Together with the initial universal move and the preprocessing moves, the total number of rounds is:
\begin{equation*}
\lceil\log_t(n)\rceil + \log(n) + 5 \leq \frac{\log(n)}{\log(t)} + \log(n) + 6 = \log(n)\left(1 + \frac{1}{\log(t)}\right) + 6 \leq (1 + \varepsilon)\log(n) + 6.
\end{equation*}
Note that $N(t)$ depends only on $t$, which in turn depends only on $\varepsilon$. So when $n < N(t)$, the number of quantifiers can be absorbed directly into the $O_\varepsilon(1)$ additive term.
\end{proof}

\def\offset{0.6}

\begin{figure}[ht]
\centering
\begin{tikzpicture}

\node [BLANK] at (-1,5*\offset) {1};
\node [BLANK] at (-0.5,5*\offset) {0};
\node [BLANK] at (0,5*\offset) {0};
\node [BLANK] at (0.5,5*\offset) {1};
\node [BLANK] at (1,5*\offset) {1};

\node [BLANK] at (-1,4*\offset) {1};
\node [BLANK] at (-0.5,4*\offset) {1};
\node [BLANK] at (0,4*\offset) {0};
\node [BLANK] at (0.5,4*\offset) {0};
\node [BLANK] at (1,4*\offset) {0};

\node [BLANK] at (-1,3*\offset) {0};
\node [BLANK] at (-0.5,3*\offset) {1};
\node [BLANK] at (0,3*\offset) {1};
\node [BLANK] at (0.5,3*\offset) {0};
\node [BLANK] at (1,3*\offset) {1};

\node [BLANK] at (-1,2*\offset) {0};
\node [BLANK] at (-0.5,2*\offset) {1};
\node [BLANK] at (0,2*\offset) {0};
\node [BLANK] at (0.5,2*\offset) {1};
\node [BLANK] at (1,2*\offset) {0};

\node [BLANK] at (-1,\offset) {1};
\node [BLANK] at (-0.5,\offset) {1};
\node [BLANK] at (0,\offset) {1};
\node [BLANK] at (0.5,\offset) {1};
\node [BLANK] at (1,\offset) {1};

\node [BLANK] at (-1,0) {1};
\node [BLANK] at (-0.5,0) {1};
\node [BLANK] at (0,0) {0};
\node [BLANK] at (0.5,0) {1};
\node [BLANK] at (1,0) {1};

\end{tikzpicture}
\begin{tikzpicture}
\node [BLANK] at (0, 0) {};
\node [BLANK] at (0, 5*\offset) {};
\node [BLANK] at (-0.4, 0) {};
\node [BLANK] at (0.4, 0) {};
\draw[->, very thick] (-0.3,2.5*\offset) -- (0.3,2.5*\offset);
\end{tikzpicture}
\begin{tikzpicture}

\node [BLANK] at (-1,5*\offset) {1};
\node [RED] at (-0.5,5*\offset) {0};
\node [BLANK] at (0,5*\offset) {0};
\node [BLANK] at (0.5,5*\offset) {1};
\node [BLANK] at (1,5*\offset) {1};

\node [BLANK] at (-1,4*\offset) {1};
\node [RED] at (-0.5,4*\offset) {1};
\node [BLANK] at (0,4*\offset) {0};
\node [BLANK] at (0.5,4*\offset) {0};
\node [BLANK] at (1,4*\offset) {0};

\node [BLANK] at (-1,3*\offset) {0};
\node [BLANK] at (-0.5,3*\offset) {1};
\node [RED] at (0,3*\offset) {1};
\node [BLANK] at (0.5,3*\offset) {0};
\node [BLANK] at (1,3*\offset) {1};

\node [BLANK] at (-1,2*\offset) {0};
\node [BLANK] at (-0.5,2*\offset) {1};
\node [RED] at (0,2*\offset) {0};
\node [BLANK] at (0.5,2*\offset) {1};
\node [BLANK] at (1,2*\offset) {0};

\node [BLANK] at (-1,\offset) {1};
\node [BLANK] at (-0.5,\offset) {1};
\node [BLANK] at (0,\offset) {1};
\node [RED] at (0.5,\offset) {1};
\node [BLANK] at (1,\offset) {1};

\node [BLANK] at (-1,0) {1};
\node [BLANK] at (-0.5,0) {1};
\node [BLANK] at (0,0) {0};
\node [RED] at (0.5,0) {1};
\node [BLANK] at (1,0) {1};

\end{tikzpicture}
\begin{tikzpicture}
\node [BLANK] at (0, 0) {};
\node [BLANK] at (0, 5*\offset) {};
\node [BLANK] at (-0.4, 0) {};
\node [BLANK] at (0.4, 0) {};
\draw[->, very thick] (-0.3,2.5*\offset) -- (0.3,2.5*\offset);
\end{tikzpicture}
\begin{tikzpicture}

\node [BLANK] at (-1,5*\offset) {1};
\node [RED] at (-0.5,5*\offset) {0};
\node [BLUE] at (0,5*\offset) {0};
\node [BLANK] at (0.5,5*\offset) {1};
\node [BLANK] at (1,5*\offset) {1};

\node [BLANK] at (-1,4*\offset) {1};
\node [RED] at (-0.5,4*\offset) {1};
\node [BLANK] at (0,4*\offset) {0};
\node [BLUE] at (0.5,4*\offset) {0};
\node [BLANK] at (1,4*\offset) {0};

\node [BLANK] at (-1,3*\offset) {0};
\node [BLUE] at (-0.5,3*\offset) {1};
\node [RED] at (0,3*\offset) {1};
\node [BLANK] at (0.5,3*\offset) {0};
\node [BLANK] at (1,3*\offset) {1};

\node [BLANK] at (-1,2*\offset) {0};
\node [BLANK] at (-0.5,2*\offset) {1};
\node [RED] at (0,2*\offset) {0};
\node [BLUE] at (0.5,2*\offset) {1};
\node [BLANK] at (1,2*\offset) {0};

\node [BLANK] at (-1,\offset) {1};
\node [BLUE] at (-0.5,\offset) {1};
\node [BLANK] at (0,\offset) {1};
\node [RED] at (0.5,\offset) {1};
\node [BLANK] at (1,\offset) {1};

\node [BLANK] at (-1,0) {1};
\node [BLANK] at (-0.5,0) {1};
\node [BLUE] at (0,0) {0};
\node [RED] at (0.5,0) {1};
\node [BLANK] at (1,0) {1};

\end{tikzpicture}
\begin{tikzpicture}
\node [BLANK] at (0, 0) {};
\node [BLANK] at (0, 5*\offset) {};
\node [BLANK] at (-0.4, 0) {};
\node [BLANK] at (0.4, 0) {};
\draw[->, very thick] (-0.3,2.5*\offset) -- (0.3,2.5*\offset);
\end{tikzpicture}
\begin{tikzpicture}

\node [BLANK] at (-1,5*\offset) {1};
\node [RED] at (-0.5,5*\offset) {0};
\node [BLUE] at (0,5*\offset) {0};
\node [GREEN] at (0.5,5*\offset) {1};
\node [BLANK] at (1,5*\offset) {1};

\node [BLANK] at (-1,4*\offset) {1};
\node [RED] at (-0.5,4*\offset) {1};
\node [GREEN] at (0,4*\offset) {0};
\node [BLUE] at (0.5,4*\offset) {0};
\node [BLANK] at (1,4*\offset) {0};

\node [BLANK] at (-1,3*\offset) {0};
\node [BLUE] at (-0.5,3*\offset) {1};
\node [RED] at (0,3*\offset) {1};
\node [GREEN] at (0.5,3*\offset) {0};
\node [BLANK] at (1,3*\offset) {1};

\node [BLANK] at (-1,2*\offset) {0};
\node [GREEN] at (-0.5,2*\offset) {1};
\node [RED] at (0,2*\offset) {0};
\node [BLUE] at (0.5,2*\offset) {1};
\node [BLANK] at (1,2*\offset) {0};

\node [BLANK] at (-1,\offset) {1};
\node [BLUE] at (-0.5,\offset) {1};
\node [GREEN] at (0,\offset) {1};
\node [RED] at (0.5,\offset) {1};
\node [BLANK] at (1,\offset) {1};

\node [BLANK] at (-1,0) {1};
\node [GREEN] at (-0.5,0) {1};
\node [BLUE] at (0,0) {0};
\node [RED] at (0.5,0) {1};
\node [BLANK] at (1,0) {1};

\end{tikzpicture}
\caption{Illustration of the technique used by $\bS$ to partition a set of structures into isomorphism classes. Here $\bS$ plays three pebble moves to break the set of six strings into distinct isomorphism classes: $\r < \b < \g$, $\r < \g < \b$, and so on. Note that three pebbling moves suffice to give each string its own isomorphism class since $3!=6$.
}
\label{fig:isomorphismclasses}
\end{figure}
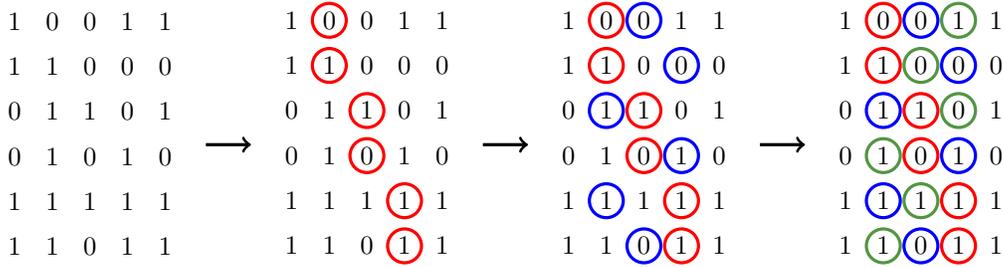

The next problem we will consider has polynomially many $n$-bit strings in $\cA$, and the remaining $n$-bit strings in $\cB$. This will correspond to our Boolean function of interest being a \emph{sparse} function. Note that this immediately implies Theorem \ref{thm-B} in Section \ref{sec:intro}.

\begin{restatable}[Polynomially Many vs.~All]{theorem}{polymanyvsall}\label{thm:3+epsilon-upper}
Let $f \colon \mathbb{N} \to \mathbb{N}$ be a function satisfying $\lim_{n \rightarrow \infty} f(n) = \infty$ and $f(n) = O(n^k)$ for some constant $k$. Then, for all $n$, and for every $\varepsilon > 0$, it is possible to separate each set of $f(n)$ $n$-bit strings from all other $n$-bit strings by a sentence with $(1 + \varepsilon)\log(n) + O_{k, \varepsilon}(1)$ quantifiers.
\end{restatable}

\begin{proof}[Proof Sketch (see Appendix \ref{apx:strings} for full proof)]
Assume $n > 2$, and pick a sufficiently large constant $k$ such that $f(n) \leq  n^k$ for all $n$. Next, pick $\varepsilon > 0$. Let $t \geq 4$ be a large enough integer so that $t \geq 2^{2k/\varepsilon}$. By Lemma \ref{lem:stirling}, we know there is some integer $N(t)$, such that for all $n \geq N(t)$, we have $\lceil\log_t(f(n))\rceil! \geq f(n)$. $\bS$ once again plays $\lceil\log_t(f(n))\rceil$ existential moves, separating the $f(n)$ strings in $\cA$ into distinct isomorphism classes by using different permutations. Now, as in the proof of Theorem \ref{thm:one-vs-all-n}, $\bS$ has reduced the games to $f(n)$ parallel \textbf{one-vs-all} string separation instances. So now, using Theorem \ref{thm:one-vs-all-n}, he can win these instances in parallel, using $(1 + \varepsilon/2)\log(n) + 6$ quantifiers for all $n \geq \max(N(t), N'(\varepsilon))$, for some $N'(\varepsilon)$ depending only on $\varepsilon$. The total number of rounds used by $\bS$ is:
\begin{align*}
\lceil\log_t(f(n))\rceil + (1 + \varepsilon/2)\log(n) + 6 & \leq (1 + \varepsilon/2)\log(n) + k\log_t(n) + 7 \\
& = (1 + \varepsilon/2)\log(n) + \frac{k\log(n)}{2k/\varepsilon} + 7 \\
& \leq (1 + \varepsilon)\log(n) + 7.
\end{align*}
Again, $N(t)$ depends only on $t$, which depends only on $k$ and $\varepsilon$, whereas $N'(\varepsilon)$ depends only on $\varepsilon$. So when $n < \max(N(t), N'(\varepsilon))$, the number of quantifiers can be absorbed into an additive term that depends only on $k$ and $\varepsilon$, giving us the $O_{k, \varepsilon}(1)$ term.
\end{proof}

Our final results concern separating arbitrary sets of $n$-bit strings from their complements. As discussed in Section \ref{sec:intro}, this corresponds exactly to defining arbitrary Boolean functions. Note that this will immediately imply Theorem \ref{thm-A} in Section \ref{sec:intro}.

\begin{restatable}[Arbitrary vs.~Arbitrary --- \emph{Upper Bound}]{theorem}{anyvsanyupper}\label{thm:all-vs-all upper bound}
For all $n$, and for every $\varepsilon > 0$, any two disjoint sets of $n$-bit strings are separable by a sentence with
$(1 + \varepsilon)\frac{n}{\log(n)} + O_\varepsilon(1)$ quantifiers.
\end{restatable}

\begin{proof}[Proof Sketch (see Appendix \ref{apx:strings} for full proof)]
We first observe that for any real number $r > 2$, $\bS$ can play $m := \lceil n/\log_r(n)\rceil$ preprocessing existential moves, putting different permutations of these $m$ pebbles on the strings in $\cA$ (i.e., the left side). A Stirling's approximation argument similar to Lemma \ref{lem:stirling} shows that there is some $N(r)$ such that for all $n \geq N(r)$, this number $m$ of preprocessing moves suffices to give each string in $\cA$ its own isomorphism class. Note that once this is done, $\bS$ has partitioned the original instance into $|\cA|$ disjoint instances of \textbf{one-vs-all} games.

Now, given $\varepsilon > 0$, we first choose $r > 2$ small enough that $\log(r) < 1 + \varepsilon/2$. $\bS$ now plays the preprocessing existential moves as described above to obtain $|\cA|$ parallel \textbf{one-vs-all} instances. Now, by Theorem \ref{thm:one-vs-all-n}, he can win these instances in parallel using Lemma \ref{lem:genparallelplay}, using $(1 + \varepsilon/2)\log(n) + 6$ rounds for all $n \geq \max(N(r), N'(\varepsilon))$, for some $N'(\varepsilon)$ depending only on $\varepsilon$. The total number of rounds needed, therefore, is:
\begin{align*}
    m + (1 + \varepsilon/2)\log(n) + 6 &\leq \frac{n}{\log(n)}\cdot\log(r) + \left(1 + \frac{\varepsilon}{2}\right)\log(n) + 7 \\
    &\leq \left(1 + \frac{\varepsilon}{2}\right)\left(\frac{n}{\log(n)} + \log(n)\right) + 7 \\
    &< (1 + \varepsilon)\frac{n}{\log(n)} + 7
\end{align*}
for all $n \geq \max(N(r), N'(\varepsilon), N''(\varepsilon))$, where for all $n \geq N''(\varepsilon)$, we have $(1 + \varepsilon/2)\log(n) < (\varepsilon/2)\frac{n}{\log(n)}$. Since each of $N(r)$, $N'(\varepsilon)$, and $N''(\varepsilon)$ depends only on $\varepsilon$, the number of quantifiers for smaller $n$ is absorbed into the $O_\varepsilon(1)$ term.
\end{proof}

Remarkably, we cannot improve the upper bound in Theorem \ref{thm:all-vs-all upper bound} by any significant amount. The following proposition establishes this by means of a counting argument, also showing that Theorem \ref{thm-A} is tight.

\begin{restatable}[Arbitrary vs.~Arbitrary --- \emph{Lower Bound}]{theorem}{anyvsanylower}\label{prop:all-vs-all lower bound}
For all sufficiently large $n$, there is a nonempty set of $n$-bit strings, $\cA \subsetneq \{0, 1\}^n$, such that every separating sentence $\varphi$ for $(\cA, \{0, 1\}^n - \cA)$ must have at least $n/\log(n)$ quantifiers.
\end{restatable}

\begin{proof}[Proof Sketch (see Appendix \ref{apx:strings} for full proof)]
If we require $k$ quantifiers to separate any instance on $n$-bit strings (for sufficiently large $n$), we can start by counting the number of pairwise inequivalent sentences that can be written with $k$ quantifiers. Such a sentence has a quantifier prefix of length at most $k$ ($\leq 2^{k+1}$ possibilities), followed by a quantifier-free part, which is a disjunction of types ($2^{k^k}$ possibilities). This puts the total number of possible such formulas to be at most $2^k\cdot 2^{2^{k\log(k)}}$. We need this number to be at least $2^{2^n} - 2$, to account for all nonempty instances of the form $(\cA, \{0, 1\}^n - \cA)$, which require pairwise inequivalent sentences to separate. Solving this shows that we need $k \geq n/\log(n)$.
\end{proof}

\section{Conclusions \& Open Problems }\label{sec:conclusion}

We obtained nontrivial quantifier upper bounds with matching lower bounds (up to $(1 + \varepsilon)$ factors) for a variety of string separation problems. All our upper bounds arise as a result of using the technique of parallel play.

Throughout this work, with very few exceptions, we used MS games to obtain \textit{upper} bounds. It might seem unnecessary to exhibit upper bounds using game arguments, when it ordinarily suffices to exhibit separating sentences. However, the sentences implicitly arising from our game techniques are highly nontrivial to construct. 
In the case of QR, since taking disjunctions and conjunctions do not increase the quantifier rank, one can build up complex sentences out of simpler ones without paying any cost; we lose this convenience with QN, and therefore need more nuanced techniques, such as parallel play.

Natural directions to extend this work include the following:
\begin{enumerate}
\item It would be illuminating to understand the QN required to express particular string and graph properties. While our lower bound for the \textbf{one-vs-one} problem (Proposition \ref{prop:one-vs-one-n}) gave a pair of strings requiring $\log(n)$ quantifiers to separate, the counting argument in Proposition \ref{prop:all-vs-all lower bound} does not exhibit a \textit{specific} instance on $n$-bit strings that provably requires $n/\log(n)$ quantifiers to separate. Note that by (\ref{neils-inlcusions}), if we can find any property that requires $\omega(\log(n))$ quantifiers to capture, then that property lies outside of NL.

\item Is it possible to use our upper bound in Theorem \ref{thm:all-vs-all upper bound} to obtain Lupanov's upper bound of $(1+\varepsilon)2^n/\log(n)$ on the minimum formula size needed separate two sets in propositional logic (or vice versa)?

\item It is known for ordered structures that with $O(\log n)$ quantifiers, one can express the $\bit$ predicate, or equivalently, all standard arithmetic operations on elements of the universe  \cite{IMMERMAN:1999}. In particular, with $\bit$, some properties that would otherwise require $\log(n)$ quantifiers can be expressed using $O(\log(n)/\log\log(n))$ quantifiers. Understanding the use of $\bit$ and other numeric relations would be valuable.
\end{enumerate}

\section*{Acknowledgements}

The authors acknowledge Ryan Williams for numerous helpful discussions and conversations, Sebastian Pfau for an observation that improved the statement of the Parallel Play Lemma, and the anonymous reviewers for comments and suggestions that improved the quality of this manuscript. Rik Sengupta was supported by NSF CCF-1934846.

\newpage

\bibliographystyle{plainurl}
\bibliography{bibliography}

\newpage

\appendix

\section{Technical Content from Section \ref{sec:linearorders}}\label{apx:linearorders}

\begin{proposition}\label{prop:irreduciblepatterns}
The irreducible games are winnable by $\bS$ with the patterns asserted at the start of Section \ref{sec:cmaformal}.
\end{proposition}

\begin{proof}
We consider the irreducible games one by one.
\begin{enumerate}
    \item The game $\textrm{MSL}_{\forall, 1}(1)$ is winnable; $\bS$ makes a universal move by playing on any element other than $\mathsf{min}$ and $\mathsf{max}$ on each board on the right. There is no valid response by $\bD$ on the single board on the left. The pattern is $(\forall)$.
    \item The game $\textrm{MSL}_{\exists, 2}(1)$ is winnable; $\bS$ makes a dummy existential move (by playing as a matter of convention on $\mathsf{min}$), and then reverts to the strategy above for $\textrm{MSL}_{\forall, 1}(1)$ for his second move. The pattern is $(\exists, \forall)$.
    \item The game $\textrm{MSL}_{\forall, 2}(2)$ is winnable; $\bS$ makes two successive universal moves by playing on two arbitrary distinct elements other than $\mathsf{min}$ and $\mathsf{max}$ on each board on the right. $\bD$ cannot match this on the boards on the left. We remark that $\textrm{MSL}_{\forall, 2}(2)$ is not winnable by $\bS$ if he plays in any other fashion. The pattern is $(\forall, \forall)$.
    \item The game $\textrm{MSL}_{\forall, 3}(2)$ is winnable; $\bS$ follows the same strategy as in $\textrm{MSL}_{\forall, 2}(2)$ in rounds $1$ and $3$, except that he makes a dummy existential move in round $2$ (by playing as a matter of convention on $\mathsf{min}$). The pattern is $(\forall, \exists, \forall)$.
\end{enumerate}
This concludes the proof.
\end{proof}

\splitrules*

\begin{proof} \label{app:lo-proofs}
We prove Claim \ref{claim:splitrules} in cases (i) and (iii). The other two cases are similar. Note that our analysis of $\textrm{MSL}_{\exists,4}(5)$ was an example of case (ii).

Consider case (i). Figure \ref{parallel path case1 fig} shows the gameplay through to the configuration immediately after $\bS$'s round $2$ move. Once $\bD$ makes her oblivious response, we can discard some of the boards following Observation \ref{obs:discard}, and note that the game splits into two games: $(\cA_1, \cB_1)$ corresponding to the isomorphism class $\mathsf{min} < \r < \b < \mathsf{max}$, and $(\cA_2, \cB_2)$ corresponding to the isomorphism class $\mathsf{min} < \b < \r < \mathsf{max}$. The game proceeds within the linear orders $L[\r, \mathsf{max}]$ in $(\cA_1, \cB_1)$, and within the linear orders $L[\mathsf{min}, \r]$ in $(\cA_2, \cB_2)$, using Observation \ref{obs:discard} to discard any responses outside those ranges. By construction, these both correspond to $\textrm{MSL}_{\forall,k-1}(\ell)$ games.

\begin{figure*}[ht]
\begin{center}
\tikzset{Tleft/.style={ellipse,draw=black, inner sep=0pt, minimum size=.9cc,
    line width=1pt,fill=orange!20}}
\tikzset{TRleft/.style={ellipse,draw=red, inner sep=0pt, minimum size=.9cc,
    line width=2pt,fill=orange!20}}
\tikzset{Tright/.style={ellipse,draw=black, inner sep=0pt, minimum size=.9cc,
    line width=1pt,fill=indigo!20}}
  \tikzset{Red/.style={circle,draw=red, inner sep=0pt, minimum size=.9cc, line width=2pt}}
\tikzset{Blue/.style={circle,draw=blue, inner sep=0pt, minimum size=.9cc, line width=2pt}}
\tikzset{Green/.style={circle,draw=dg, inner sep=0pt, minimum size=.9cc, line width=2pt}}
\tikzset{Black/.style={circle,draw=black, inner sep=0pt, minimum size=.9cc, line width=1pt}}
\tikzset{TreeNode/.style={rectangle,draw=black, inner sep=0pt, minimum size=1.3cc, line width=1pt}}
\tikzset{BigTreeNode/.style={rectangle,draw=black, inner sep=0pt, minimum size=1.8cc, line width=1pt}}

\begin{tikzpicture}[scale=.09]
\node at (-38,50) {$L_{\leq 2\ell}$};
\node at (42,50) {$L_{> 2\ell}$};

\node [TreeNode] (T0) at(0,40)  {};
   \node at (T0) {$\exists \r$};
  \node [TreeNode] (T1) at (0,24) {$\forall \b$};
  \draw[line width=1pt,->,color=black] (T0) -- (T1);

\node at (-52,42.25) {$\leq \ell$};
\node at (-24,42.25) {$\leq \ell$};
\node at (56,42) {$>\ell$};
\node [Tleft] (1a1) at (-66,40) {{\scriptsize \mn}};
\node [Red] (1amid) at (-38,40) {{\scriptsize $\r$}};
\node [Tright] (1a8) at (-10,40) {{\scriptsize \mx}};
\node [Tright] (1b9) at (74,40) {{\scriptsize \mx}};
\node [Red]   (1b2) at (38,40) {{\scriptsize $\r$}};
\node [Tleft] (1b1) at (10,40) {{\scriptsize \mn}};
\foreach \from/\to in {1a1/1amid,1amid/1a8,1b1/1b2,1b2/1b9}
\draw[-,line width=1pt,color=black] (\from) -- (\to);

\node at (28,36) {$>\ell$};
\node [Tright] (1c9) at (74,34) {{\scriptsize \mx}};
\node [Red]   (1c2) at (46,34) {{\scriptsize $\r$}};
\node [Tleft] (1c1) at (10,34) {{\scriptsize \mn}};
\foreach \from/\to in {1a1/1amid,1amid/1a8,1c1/1c2,1c2/1c9}
\draw[-,line width=1pt,color=black] (\from) -- (\to);

\node [Tright] (2b9) at (74,24) {{\scriptsize \mx}};
\node [Red]   (2b2) at (38,24) {{\scriptsize $\r$}};
\node [Blue]   (2b3) at (56,24) {{\scriptsize $\b$}};
\node [Tleft] (2b1) at (10,24) {{\scriptsize \mn}};

\node [Tright] (2c9) at (74,18) {{\scriptsize \mx}};
\node [Red]   (2c2) at (46,18) {{\scriptsize $\r$}};
\node [Blue]  (2c3) at (28,18) {{\scriptsize $\b$}};
\node [Tleft] (2c1) at (10,18) {{\scriptsize \mn}};
\foreach \from/\to in {2b1/2b2,2b2/2b3,2b3/2b9,2c1/2c3,2c3/2c2,2c2/2c9}
\draw[dotted,line width=1pt,color=black] (\from) -- (\to);
\end{tikzpicture}
\end{center}
\caption{The first round and a half of the $\textrm{MSL}_{\exists,k}(2\ell)$ game
  according to the $\mathsf{CMA}$ strategy.}
\label{parallel path case1 fig}
\end{figure*}
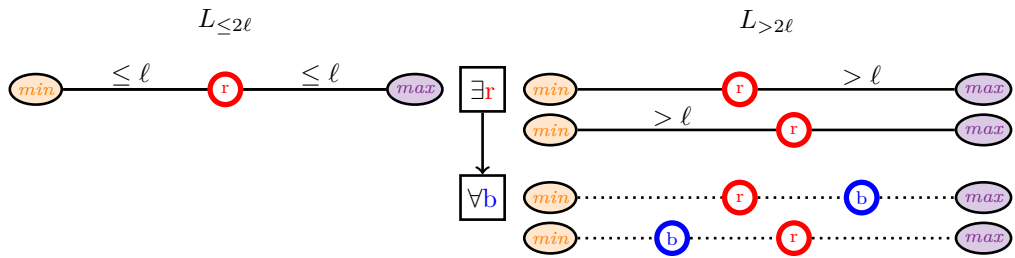

Now consider case (iii). Figure \ref{path thm case2 fig} shows the gameplay through to the configuration immediately after $\bS$'s round $2$ move. Once $\bD$ makes her oblivious response, we can discard some of the boards following Observation \ref{obs:discard}, and note that the game splits into two games: $(\cA_1, \cB_1)$ corresponding to the isomorphism class $\mathsf{min} < \b < \r < \mathsf{max}$, and $(\cA_2, \cB_2)$ corresponding to the isomorphism class $\mathsf{min} < \r < \b < \mathsf{max}$. The game proceeds within the linear orders $L[\mathsf{min}, \r]$ in $(\cA_1, \cB_1)$, and within the linear orders $L[\r, \mathsf{max}]$ in $(\cA_2, \cB_2)$, using Observation \ref{obs:discard} to discard any responses outside those ranges. By construction, these correspond to an $\textrm{MSL}_{\exists,k-1}(\ell - 1)$ game and an $\textrm{MSL}_{\exists,k-1}(\ell)$ respectively.

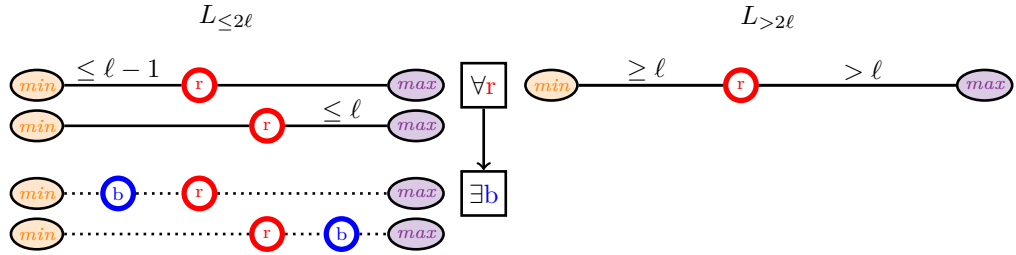
\begin{figure*}[ht]
\begin{center}
\tikzset{Tleft/.style={ellipse,draw=black, inner sep=0pt, minimum size=.9cc,
    line width=1pt,fill=orange!20}}
\tikzset{TRleft/.style={ellipse,draw=red, inner sep=0pt, minimum size=.9cc,
    line width=2pt,fill=orange!20}}
\tikzset{Tright/.style={ellipse,draw=black, inner sep=0pt, minimum size=.9cc,
    line width=1pt,fill=indigo!20}}
\tikzset{Red/.style={circle,draw=red, inner sep=0pt, minimum size=.9cc, line width=2pt}}
\tikzset{Blue/.style={circle,draw=blue, inner sep=0pt, minimum size=.9cc, line width=2pt}}
\tikzset{Green/.style={circle,draw=dg, inner sep=0pt, minimum size=.9cc, line width=2pt}}
\tikzset{Black/.style={circle,draw=black, inner sep=0pt, minimum size=.9cc, line width=1pt}}
\tikzset{TreeNode/.style={rectangle,draw=black, inner sep=0pt, minimum size=1.3cc, line width=1pt}}
\tikzset{BigTreeNode/.style={rectangle,draw=black, inner sep=0pt, minimum size=1.8cc, line width=1pt}}

\begin{tikzpicture}[scale=.09]
\node at (-38,50) {$L_{\leq 2\ell}$};
\node at (42,50) {$L_{> 2\ell}$};
\node [TreeNode] (T0) at(0,40)  {};

   \node at (T0) {$\forall \r$};
  \node [TreeNode] (T1) at (0,24) {$\exists \b$};
  \draw[line width=1pt,->,color=black] (T0) -- (T1);

\node at (56,42.25) {$>\ell$};
\node at (24,42.25) {$\geq\ell$};
\node at (-54,42.25) {$\leq \ell-1$};
\node [Tright] (1a1) at (74,40) {{\scriptsize \mx}};
\node [Red] (1amid) at (38,40) {{\scriptsize $\r$}};
\node [Tleft] (1a8) at (10,40) {{\scriptsize \mn}};
\node [Tright] (1b9) at (-10,40) {{\scriptsize \mx}};
\node [Red]   (1b2) at (-42,40) {{\scriptsize $\r$}};
\node [Tleft] (1b1) at (-66,40) {{\scriptsize \mn}};
\foreach \from/\to in {1a1/1amid,1amid/1a8,1b1/1b2,1b2/1b9}
\draw[-,line width=1pt,color=black] (\from) -- (\to);

\node at (-21,36.25) {$\leq\ell$};
\node [Tright] (1c9) at (-10,34) {{\scriptsize \mx}};
\node [Red]   (1c2) at (-32,34) {{\scriptsize $\r$}};
\node [Tleft] (1c1) at (-66,34) {{\scriptsize \mn}};
\foreach \from/\to in {1a1/1amid,1amid/1a8,1c1/1c2,1c2/1c9}
\draw[-,line width=1pt,color=black] (\from) -- (\to);

\node [Tright] (2b9) at (-10,24) {{\scriptsize \mx}};
\node [Red]   (2b2) at (-42,24) {{\scriptsize $\r$}};
\node [Blue]   (2b3) at (-54,24) {{\scriptsize $\b$}};
\node [Tleft] (2b1) at (-66,24) {{\scriptsize \mn}};

\node [Tright] (2c9) at (-10,18) {{\scriptsize \mx}};
\node [Red]   (2c2) at (-32,18) {{\scriptsize $\r$}};
\node [Blue]  (2c3) at (-21,18) {{\scriptsize $\b$}};
\node [Tleft] (2c1) at (-66,18) {{\scriptsize \mn}};
\foreach \from/\to in {2b1/2b3,2b3/2b2,2b2/2b9,2c1/2c2,2c2/2c3,2c3/2c9}
\draw[dotted,line width=1pt,color=black] (\from) -- (\to);
\end{tikzpicture}
\end{center}
\caption{The first round and a half of the $\textrm{MSL}_{\forall,k}(2\ell)$ game
  according to the $\mathsf{CMA}$ strategy.}

\label{path thm case2 fig}
\end{figure*}

\end{proof}

\cmawelldefined*

\begin{proof}
We first make a simple claim.

\begin{claim} \label{claim:simple}
Suppose $\cA$ and $\cB$ are two sets of pebbled structures that can be partitioned into $\cA = \cA_1 \cup \cA_2$ and $\cB = \cB_1 \cup \cB_2$, such that no pebbled structure in $\cA_i$ forms a matching pair with a pebbled structure in $\cB_j$, for $i \neq j$. Suppose that $\bS$ has a strategy $\cS_1$ to win the $(r+3)$-round MS-game on $(\cA_1, \cB_1)$ with pattern $\mathsf{pat}(\cS_1) = (Q_1, \ldots, Q_r, \forall, \exists, \forall)$, and also has a strategy $\cS_2$ to win the $(r+2)$-round MS-game on $(\cA_2, \cB_2)$ with pattern $\mathsf{pat}(\cS_2) = (Q_1, \ldots, Q_k, \forall, \forall)$. Then, $\bS$ has a strategy $\cS$ to win the $(r+3)$-round MS-game on $(\cA, \cB)$ satisfying $\mathsf{pat}(\cS) = \mathsf{pat}(\cS_1)$.
\end{claim}
\begin{proof}[Proof of Claim \ref{claim:simple}]
To see this, note that for $\cS$, $\bS$ simply plays the two subgames $(\cA_1, \cB_1)$ and $(\cA_2, \cB_2)$ in parallel using Lemma \ref{lem:genparallelplay}, according to the common prefix of the pattern, for the first $r+1$ rounds. In round $r + 2$, he plays according to $\cS_1$ in $(\cA_1, \cB_1)$, and makes arbitrary dummy existential moves in $(\cA_2, \cB_2)$. In round $r + 3$, he then finishes both subgames according to the final universal move dictated by $\cS_1$ in $(\cA_1, \cB_1)$, and $\cS_2$ in $(\cA_2, \cB_2)$.
\end{proof}
Back to the proof of the main lemma. Assume inductively that the lemma holds for all MSP games up to length $\ell$. Assume additionally that the pattern of both subgames alternate, either ending in one or two universal quantifiers, and that the longer subgame (if there is one) takes at most one more round than the shorter one. We shall show that both the lemma and this assumption continue to hold for the game of length $\ell + 1$.

There are four cases depending on whether we are considering an $\exists$ or $\forall$ game, and whether $\ell$ is even or odd.

\medskip

\noindent \textbf{Case 1: $\exists$ game, $\ell$ odd.} The game under consideration is $\textrm{MSL}_{\exists, r}(\ell+1)$. Since $\ell + 1$ is even, we split into two subgames that are both $\textrm{MSL}_{\forall, r-1}((\ell+1)/2)$, which clearly have the same pattern.

\smallskip

\noindent \textbf{Case 2: $\exists$ game, $\ell$ even.} The game under consideration is $\textrm{MSL}_{\exists, r}(\ell+1)$. Since $\ell+1$ is odd, we split into the subgames $\textrm{MSL}_{\forall, r-1}(\ell/2)$ and $\textrm{MSL}_{\forall, r-1}(\ell/2 + 1)$. If the subgames have strategies with the same number of moves, then they both must have the pattern $(\forall, \exists, \ldots, \forall, \exists, \forall)$ or $(\forall, \exists, \ldots, \forall, \forall)$. In this case, the lemma holds. Otherwise, the strategy for the $\textrm{MSL}_{\forall, r - 1}(\ell/2 + 1)$ game has one more round than the other game (by assumption), and then the two patterns for the subgames can line up in one of the two following ways (in each case with the longer pattern on top):
\begin{align*}
    &\qquad\qquad (\forall, \exists, \cdots, \forall, \exists, \forall) && (\forall, \exists, \cdots, \forall, \exists, \forall, \forall) \\
    &\qquad\qquad (\forall, \exists, \cdots, \forall, \forall) && (\forall, \exists, \cdots, \forall, \exists, \forall)
\end{align*}
In the first case, we are done by Claim \ref{claim:simple}. In the second case, $\bS$ just plays an arbitrary universal move at the end of the second subgame, and the lemma once again follows.

\medskip

\noindent \textbf{Case 3: $\forall$ game, $\ell$ odd.} The game under consideration is $\textrm{MSL}_{\forall, r}(\ell+1)$. Since $\ell+1$ is even, we split into the subgames $\textrm{MSL}_{\exists, r-1}((\ell+1)/2)$ and $\textrm{MSL}_{\exists, r-1}((\ell+1)/2 - 1)$. The analysis is now very similar to Case 2. If the subgames have strategies with the same number of moves, then they both have the pattern $(\exists, \forall, \ldots, \exists, \forall)$ or $(\exists, \forall, \ldots, \exists, \forall, \forall)$, and the lemma follows. Otherwise, we have the possibilities:
\begin{align*}
    &\qquad\qquad (\exists, \forall, \cdots, \exists, \forall, \forall) && (\exists, \forall, \cdots, \exists, \forall, \exists, \forall) \\
    &\qquad\qquad (\exists, \forall, \cdots, \exists, \forall) && (\exists, \forall, \cdots, \exists, \forall, \forall)
\end{align*}
And the analysis of these two cases is just like in Case 2, whereby the lemma follows.

\smallskip

\noindent \textbf{Case 4:  $\forall$ game, $\ell$ even.} The game under consideration is $\textrm{MSL}_{\forall, r}(\ell+1)$. Since $\ell+1$ is odd, we split into two subgames of $\textrm{MSL}_{\exists, r-1}(\ell/2)$, which have the same pattern.

\medskip

The requisite alternation pattern is clearly maintained, from the definition of the strategy. The assumption that the longer subgame takes at most one more round than the shorter subgame follows by noting that the lengths of the patterns are monotonic in $\ell$ and never increase by more than one.
\end{proof}

\begin{restatable}{proposition}{qpowersoftwo}\label{prop:q*-powers-of-two}
For all $k \geq 1$, we have $q^*_\forall(2^k) = q^*_\forall(2^{k-1}) + 1$. Similarly, for all $k \geq 2$, we have $q^*_\exists(2^k) = q^*_\exists(2^{k-1}) + 1$.
\end{restatable}

\begin{proof}
We start by showing a claim.

\begin{claim}\label{claim:q*inductionseparate} 
We have $q^*_\exists(\ell ) = 2 + q^*_\exists(\lfloor(\ell +1)/4\rfloor)$ for all $\ell \geq 5$. Similarly, we have $q^*_\forall(\ell ) = 2 + q^*_\forall(\lfloor(\ell  + 2)/4\rfloor)$ for all $\ell  \geq 3$.
\end{claim}

\begin{proof}
It follows directly from Lemma \ref{lem:q*-vals} that $q^*_\exists(\ell) = q^*_\forall(\lceil\ell/2\rceil) + 1$ for $\ell \geq 2$. Similarly, $q^*_\forall(\ell) = q^*_\exists(\lfloor\ell/2\rfloor) + 1$ for $\ell \geq 3$. It follows that for $\ell \geq 1$:
    \begin{equation}
        q^*_\exists(4\ell + 1) = q^*_\exists(4\ell + 2) = q^*_\forall(2\ell + 1) + 1 = q^*_\exists(\ell) + 2. \label{enum1}
    \end{equation}
Similarly, for $\ell \geq 2$:
    \begin{equation}
        q^*_\exists(4\ell - 1) = q^*_\exists(4\ell) = q^*_\forall(2\ell) + 1 = q^*_\exists(\ell) + 2. \label{enum2}
    \end{equation}
Combining these gives us the result (and the associated ranges).
\end{proof}
We now go back to the proof of Proposition \ref{prop:q*-powers-of-two}. We prove both statements by induction on $k$. For the first, we have $q^*_\forall(1) = 1$ and $q^*_\forall(2) = 2$, establishing the base case. Inductively, by Claim \ref{claim:q*inductionseparate}, we have:
\begin{equation*}
    q^*_\forall(2^k) = q^*_\forall\left(\left\lfloor\frac{2^k  + 2}{4}\right\rfloor\right) + 2 = q^*_\forall(2^{k-2}) + 2 = q^*_\forall(2^{k-1}) + 1,
\end{equation*}
with the last equality following from the induction hypothesis. Similarly, for the second part, we start with $q^*_\exists(2) = 2$ and $q^*_\exists(4) = 3$, establishing the base case. Inductively, by Claim \ref{claim:q*inductionseparate}, we have:
\begin{equation*}
    q^*_\exists(2^k) = q^*_\exists\left(\left\lfloor\frac{2^k  + 1}{4}\right\rfloor\right) + 2 = q^*_\exists(2^{k-2}) + 2 = q^*_\exists(2^{k-1}) + 1.\qedhere
\end{equation*}
\end{proof}

\begin{remark}\label{rem:recursive++}
Lemma \ref{lem:q*-vals} and Claim \ref{claim:q*inductionseparate} are more than just recursive expressions for $q^*_\exists(\ell)$ and $q^*_\forall(\ell)$. By virtue of Lemma \ref{lem:cma-is-well-defined}, we can now read off a quantifier prefix establishing $q^*_\exists(2\ell)$ in terms of $q^*_\forall(\ell)$, and analogously for the other expressions.
\end{remark}

\alternationone*

\begin{proof}
The theorem is certainly true for small values of $\ell$; e.g., when $\ell = 1$, $q^*(1) = q^*_\forall(1)$, and the sentence corresponding to that strategy has quantifier prefix $\forall$. Similarly, when $\ell = 2$, $q^*(1) = q^*_\exists(1)$, and the sentence corresponding to that strategy has quantifier prefix $\exists\forall$. The theorem can be verified for $\ell \leq 5$ simply by referring to Table \ref{q*-table}. We now proceed by induction.

Suppose $\ell$ is even, say $\ell = 2\ell'$. There are three cases:
\begin{itemize}
    \item If $q^*_\exists(\ell) < q^*_\forall(\ell)$, this means by Lemma \ref{lem:q*-vals} that $q^*_\forall(\ell') < q^*_\exists(\ell')$. So, by induction, there is a separating sentence $\sigma_{\ell'}$ with quantifier prefix $\forall\exists\ldots\forall$ consisting of $q^*_\forall(\ell')$ alternating quantifiers. But by Lemma \ref{lem:q*-vals} and Remark \ref{rem:recursive++}, we can obtain a separating sentence $\sigma_\ell$ with quantifier prefix $\exists\forall\ldots\forall$ consisting of $q^*_\exists(\ell)$ alternating quantifiers.
    \item If $q^*_\forall(\ell) < q^*_\exists(\ell)$, this means by Lemma \ref{lem:q*-vals} that $q^*_\exists(\ell') < q^*_\forall(\ell')$. Again, by induction, there is a separating sentence $\sigma_{\ell'}$ with quantifier prefix $\exists\forall\ldots\forall$ consisting of $q^*_\exists(\ell')$ alternating quantifiers. By Lemma \ref{lem:q*-vals}, we can obtain a separating sentence $\sigma_\ell$ with quantifier prefix $\forall\exists\ldots\forall$ consisting of $q^*_\forall(\ell)$ alternating quantifiers.
    \item If $q^*_\forall(\ell) = q^*_\exists(\ell)$, this means by Lemma \ref{lem:q*-vals} that $q^*_\exists(\ell') = q^*_\forall(\ell')$. Again, by induction, there is a separating sentence $\sigma_{\ell'}$ consisting of $q^*(\ell')$ alternating quantifiers ending with a $\forall$. By Lemma \ref{lem:q*-vals}, we can obtain a separating sentence $\sigma_\ell$ by prepending a quantifier to $\sigma_{\ell'}$ maintaining alternation. This would still contain $q^*(\ell)$ alternating quantifiers.
\end{itemize}

Now suppose $\ell$ is odd, say $\ell = 2\ell' + 1$. There are three cases:
\begin{itemize}
    \item If $q^*_\exists(\ell) < q^*_\forall(\ell)$, this means by Lemma \ref{lem:q*-vals} that $q^*_\forall(\ell' + 1) < q^*_\exists(\ell') \leq q^*_\exists(\ell' + 1)$. By induction, there is a separating sentence $\sigma_{\ell' + 1}$ with quantifier prefix $\forall\exists\ldots\forall$ consisting of $q^*_\forall(\ell' + 1)$ alternating quantifiers. Then we can prepend a $\exists$ to obtain a separating sentence $\sigma_\ell$ with quantifier prefix consisting of $q^*_\exists(\ell)$ alternating quantifiers.
    \item If $q^*_\forall(\ell) < q^*_\exists(\ell)$, this means by Lemma \ref{lem:q*-vals} that $q^*_\exists(\ell') < q^*_\forall(\ell' + 1)$. If $q^*_\exists(\ell') < q^*_\forall(\ell')$, we are again done by induction. If $q^*_\exists(\ell') = q^*_\forall(\ell')$, this means $q^*_\forall(\ell' + 1) > q^*_\forall(\ell')$, implying by Claim \ref{claim:q*inductionseparate} that $\ell' \equiv 1\pmod 4$. But then $q^*_\exists(\ell' + 1) = q^*_\exists(\ell') < q^*_\forall(\ell' + 1)$. Therefore, $q^*(\ell') = q^*(\ell' + 1) = q^*_\exists(\ell' + 1)$, and so any alternating quantifier prefix with $q^*(\ell')$ quantifiers ending with a $\forall$ must start with a $\exists$. Since by induction, $\sigma_{\ell'}$ has $q^*(\ell')$ alternating quantifiers ending with a $\forall$, it must also start with a $\exists$. Now again, we are done by prepending a $\forall$, by induction.
    \item If $q^*_\forall(\ell) = q^*_\exists(\ell)$, this means by Lemma \ref{lem:q*-vals} that $q^*_\exists(\ell') = q^*_\forall(\ell' + 1)$. Again, by induction, if $q^*_\forall(\ell' + 1) < q^*_\exists(\ell' + 1)$, we are done. If $q^*_\forall(\ell' + 1) = q^*_\exists(\ell' + 1)$, however, we have to be a little more careful. In that situation, if $q^*_\exists(\ell') = q^*_\forall(\ell')$, then $q^*(\ell') = q^*(\ell' + 1)$, and then the sentences $\sigma_{\ell'}$ and $\sigma_{\ell' + 1}$ have the same quantifier prefix. Depending on the leading quantifier in that prefix, we can inductively use either $q^*_\forall(\ell)$ or $q^*_\exists(\ell)$. Otherwise, $q^*_\exists(\ell') > q^*_\forall(\ell')$. But then, $\sigma_{\ell'}$ starts with a $\forall$, and therefore, the sentence $\sigma'_{\ell'}$ that is used by $\bS$ in the $q^*_\exists(\ell')$ strategy has $q^*(\ell') + 1 = q^*(\ell' + 1)$ quantifiers and starts with an $\exists$. Now, by induction, we can obtain a sentence $\sigma_\ell$ using $q^*_\forall(\ell)$ that calls $\sigma'_{\ell'}$, and has $1 + q^*(\ell' + 1)$ alternating quantifiers ending with a $\forall$. Since $q^*(\ell' + 1) = q^*_\forall(\ell' + 1) = q^*_\exists(\ell') = q^*(\ell) - 1$, we are done.\qedhere
\end{itemize}
\end{proof}

\alternationtwo*

\begin{proof}
As in the proof sketch, when $\ell = 1$, the theorem follows directly from Theorem \ref{thm:alternation1}. So suppose $\ell > 1$.

Again let $\cA = \{L_\ell\}$, and let $\cB = \cB_1 \sqcup \cB_2$, where $\cB_1 = L_{\leq \ell - 1}$, and $\cB_2 = L_{> \ell}$. By Theorem \ref{thm:alternation1}, there is a sentence $\sigma_{\leq \ell}$ that is true for $L_{\leq \ell}$ and false for $L_{> \ell} = \cB_2$, with the given alternating quantifier prefix, with $q^*(\ell)$ quantifiers. Similarly, there is a sentence $\sigma_{\leq \ell - 1}$ which is true for $L_{\leq \ell - 1} = \cB_1$ and false for $L_{\geq \ell}$, with the given alternating quantifier prefix, with $q^*(\ell - 1) \leq q^*(\ell)$ quantifiers. Assume these two sentences both have $q^*(\ell)$ quantifiers (possibly by prepending a dummy leading quantifier to $\sigma_{\leq \ell - 1}$). Let $\sigma_2 := \sigma_{\leq \ell}$ and $\sigma_1 := \lnot\sigma_{\leq \ell - 1}$. Note that $\sigma_1$ separates $(\cA, \cB_1)$ (say with strategy $\cS_1$), and $\sigma_2$ separates $(\cA, \cB_2)$ (say with strategy $\cS_2$), and so $\sigma_1 \land \sigma_2$ separates $(\cA, \cB)$. Furthermore, $\sigma_1$ and $\sigma_2$ both have alternating quantifier prefixes of the same length $q^*(\ell)$, but they are complements of each other: $\sigma_2$ ends in a $\forall$, and $\sigma_1$ ends in a $\exists$.

Consider the MS game on $(\cA, \cB)$. We will now give a $\bS$ strategy. $\bS$ always starts with a universal move. Exactly one of the sentences $\sigma_1$ and $\sigma_2$ begins with a $\forall$.

If the sentence with a leading $\forall$ is $\sigma_1$, $\bS$ plays his round $1$ moves, playing pebble $\r$ on the element $\mathsf{max}$ in all boards in $\cB_2$, and according to the strategy $\cS_1$ in all boards in $\cB_1$. Note that, by virtue of $\cS_1$ being the $\mathsf{CMA}$ strategy, $\bS$ never plays $\r$ on the element $\mathsf{max}$ in any board in $\cB_1$. Therefore, every board in $\cB_1$ satisfies $\r \neq \mathsf{max}$, whereas every board in $\cB_2$ satisfies $\r = \mathsf{max}$. Once $\bD$ has responded obliviously, partition $\cA$ as $\cA_1 \sqcup \cA_2$ such that every board in $\cA_1$ satisfies $\r \neq \mathsf{max}$, whereas every board in $\cA_2$ satisfies $\r = \mathsf{max}$ as well.

Now, the sub-games $(\cA_1, \cB_1)$ and $(\cA_2, \cB_2)$ can be played in parallel; there will be no matching pair from $\cA_i$ and $\cB_j$ for $i \neq j$; furthermore, the two strategies both have patterns that are subsequences of the sequence $(\forall, \exists, \ldots, \exists, \forall)$, which has length $q^*(\ell) + 1$ or $q^*(\ell) + 2$ depending on the parity of $q^*(\ell)$. Therefore, by Lemma \ref{lem:genparallelplay}, the result follows.
\end{proof}

\section{Technical Content from Section \ref{sec:strings}}\label{apx:strings}

\stirling*

\begin{proof}
    Since $\lim_{n\to\infty}f(n) = \infty$, there is some $N(t)$ such that $f(n) \geq t^{et}$ for all $n \geq N(t)$, where $e$ is the base of the natural logarithm. By Stirling's formula, we have:
    \begin{equation*}
    \lceil \log_t(f(n)) \rceil ! \geq \left(\left(\frac{\lceil \log_t(t^{et}) \rceil}{e}\right)^{\lceil \log_t(f(n)) \rceil}\right)
    \geq \left(t^{\lceil \log_t(f(n)) \rceil}\right)
    \geq f(n),
\end{equation*}
where we have used $f(n) \geq t^{et}$ in the first inequality.
\end{proof}

\onevsalln*

\begin{proof}
Fix any $\varepsilon > 0$, and fix any integer $t \geq 2^{1/\varepsilon}$. By Lemma \ref{lem:stirling}, we know there is some integer $N(t)$, such that for all $n \geq N(t)$, $\lceil\log_t(n)\rceil! \geq n$. For any such $n$, fix an arbitrary $w \in \{0, 1\}^n$, and let $\cA = \{w\}$, and $\cB = \{0, 1\}^n - \{w\}$.

Consider the MS game on $(\cA, \cB)$. Every $w' \in \cB$ differs from $w$ in at least one bit. In round $1$, $\bS$ plays a universal move, placing a pebble on each $w' \in \cB$ on an index that disagrees with $w$ at that index. Assume $\bD$ responds obliviously, so that there are $n$ resulting pebbled strings in $\cA$. For the next $\lceil\log_t(n)\rceil$ rounds, $\bS$ plays only existential moves, placing the $\lceil\log_t(n)\rceil$ pebbles in distinct permutations on the $n$ strings in $\cA$, creating $n$ distinct isomorphism classes. Note that he can do so, since $\lceil\log_t(n)\rceil \leq n$ as long as $t > 2$. Once this is done, we can discard all structures in $\cB$ that do not have one of these $n$ isomorphism classes using Observation \ref{obs:discard}. For each surviving isomorphism class $\mathfrak{C}$, there is exactly one string $w_\mathfrak{C}$ in $\cA$ in that isomorphism class; consider all strings in that same isomorphism class $\mathfrak{C}$ in $\cB$. We can discard each of these strings with its round $1$ pebble at the same index as $w_\mathfrak{C}$, using Observation \ref{obs:discard}. We are now left with many strings in $\cB$, each with its round $1$ pebble at a different index from $w_\mathfrak{C}$. Note that $\bS$ can henceforth view this as a game simply about lengths, and can employ any \textbf{one-vs-all} linear order strategy. Once we do this for each such isomorphism class $\mathfrak{C}$, we have reduced the entire game to $n$ instances of \textbf{one-vs-all} sub-games on linear orders, with no two structures from different sub-games in the same isomorphism class.

By Lemma \ref{lem:genparallelplay} and Theorem \ref{thm:alternation2}, $\bS$ can now win these parallel games in at most $\log(n) + 4$ moves. This gives a total number of moves that is at most:
\begin{equation*}
\lceil\log_t(n)\rceil + \log(n) + 5 \leq \frac{\log(n)}{\log(t)} + \log(n) + 6 = \log(n)\left(1 + \frac{1}{\log(t)}\right) + 6 \leq (1 + \varepsilon)\log(n) + 6.
\end{equation*}
Note that $N(t)$ depends only on $t$, which in turn depends only on $\varepsilon$. So when $n < N(t)$, the number of quantifiers can be absorbed directly into the $O_\varepsilon(1)$ additive term. This gives us the desired result.
\end{proof}

\polymanyvsall*

\begin{proof}
Assume $n > 2$, and pick a sufficiently large constant $k$ such that $f(n) \leq  n^k$ for all $n$. This is possible because $f(n)$ is polynomially bounded. Next, pick $\varepsilon > 0$. Let $t \geq 4$ be a large enough integer so that $t \geq 2^{2k/\varepsilon}$. By Lemma \ref{lem:stirling}, we know there is some integer $N(t)$, such that for all $n \geq N(t)$, we have $\lceil\log_t(f(n))\rceil! \geq f(n)$. So $\bS$ once again starts by playing $\lceil\log_t(f(n))\rceil$ existential moves in the first few rounds, separating the $f(n)$ strings in $\cA$ into distinct isomorphism classes by using different permutations. Note that he can do so since $t \geq 4$, and so $\lceil\log_t(f(n))\rceil \leq \log_4(f(n)) + 1 \leq \log_4(2^n) + 1  = \log_4(4^{n/2}) + 1 = n/2 + 1 < n$; with the last inequality holding because we have assumed $n > 2$. Thus, there is enough space to put the pebbles on the strings. Now, as in the proof of Theorem \ref{thm:one-vs-all-n}, $\bS$ has reduced the games to $f(n)$ parallel \textbf{one-vs-all} string separation instances.

Given $\varepsilon$, using the proof of Theorem \ref{thm:one-vs-all-n}, we know that there is some $N'(\varepsilon)$ such that for all $n \geq N'(\varepsilon)$, $\bS$ has a winning strategy on each of these instances using $(1 + \varepsilon/2)\log(n) + 6$ rounds (we have used $\varepsilon/2$ instead of $\varepsilon$ here), using the \emph{same} pattern. Therefore, for all $n \geq \max(N(t), N'(\varepsilon))$, $\bS$ can use Lemma \ref{lem:genparallelplay} to win the entire game. The total number of rounds used by $\bS$ is:
\begin{align*}
&\lceil\log_t(f(n))\rceil + (1 + \varepsilon/2)\log(n) + 6 \\
&\qquad \leq (1 + \varepsilon/2)\log(n) + k\log_t(n) + 7 \\
&\qquad \leq (1 + \varepsilon/2)\log(n) + \frac{k\log(n)}{2k/\varepsilon} + 7 \\
&\qquad = (1 + \varepsilon)\log(n) + 7.
\end{align*}
Here, $N(t)$ depends only on $t$, which depends only on $k$ and $\varepsilon$, whereas $N'(\varepsilon)$ depends only on $\varepsilon$. So when $n < \max(N(t), N'(\varepsilon))$, the number of quantifiers can be absorbed into an additive term that depends only on $k$ and $\varepsilon$, giving us the additive $O_{k, \varepsilon}(1)$ term desired.
\end{proof}

\anyvsanyupper*

\begin{proof}
We begin with a claim.

\begin{claim}\label{claim:realr}
For every real number $r > 2$, there is some $N_r \in \mathbb{N}$ such that for all $n \geq N_r$, we have $m := \lceil n/\log_r(n)\rceil$ satisfies:
\begin{itemize}
    \item $m \leq n$;
    \item $m! \geq 2^n$.
\end{itemize}
\end{claim}

\begin{proof}[Proof of Claim \ref{claim:realr}]
We first note that as long as $n > r^2$, we have $m < \lceil n/2\rceil \leq n$. On the other hand, by Stirling's approximation, we have:
\begin{equation*}
m! = \left\lceil \frac{n}{\log_r(n)} \right\rceil! > \left( \frac{n}{e\log_r(n)} \right)^{\frac{n}{\log_r(n)}}.
\end{equation*}
We wish to show that the right hand side of this equation is at least $2^n$. Therefore, taking base-$2$ logarithms, we wish to show that:
\begin{equation*} 
\frac{n}{\log_r(n)}\cdot\left(\log(n) - \log(e\log_r(n)) \right) \geq n, \text{~~ i.e., ~~}\log(n) - \log(e) - \log\log_r(n) \geq \log_r(n).
\end{equation*}
Equivalently, we need to show that:
\begin{equation*} 
    \log(n) \geq \frac{\log(n)}{\log(r)} + \log\log_r(n)  + \log(e),
\end{equation*}
or in other words:
\begin{equation*}
\log(n)\left(1 - \frac{1}{\log(r)}\right) \geq \log\log(n) - \log\log(r) + \log(e).
\end{equation*}
Because $r > 2$, we have $\log(r) > 1$, and so the left hand side above grows linearly in $\log(n)$, whereas the right hand side grows logarithmically in $\log(n)$. Hence, there is some integer $N'_r$ such that for all $n \geq N'_r$, the left hand side dominates. Therefore, setting $N_r = \max(r^2 + 1, N'_r)$ yields the result.
\end{proof}
Back to the proof of Theorem \ref{thm:all-vs-all upper bound}. The idea once again will be for $\bS$ to play enough ``preprocessing'' existential moves at the start, to give each string in $\cA$ its own isomorphism class.

Given $\varepsilon > 0$, we first choose $r > 2$ small enough that $\log(r) < 1 + \varepsilon/2$. Then, $\bS$ starts by playing the preprocessing existential moves for $m := \left\lceil n/\log_r(n)\right\rceil$ rounds as described above. By Claim \ref{claim:realr}, this is possible, and ends up with $\bS$ splitting the original instance into $|\cA|$ parallel \textbf{one-vs-all} sub-games.

Now, by the arguments in Theorem \ref{thm:one-vs-all-n}, there is some $N'(\varepsilon)$ such that for $n \geq N(\varepsilon)$, $\bS$ has a winning strategy for each of these sub-games in $(1 + \varepsilon/2)\log(n) + 6$ further rounds, using the same pattern. Therefore, using Lemma \ref{lem:genparallelplay}, he can win the entire instance with the same number of rounds by playing in parallel. The total number of rounds needed, therefore, is:
\begin{align*}
    &m + (1 + \varepsilon/2)\log(n) + 6 \\
    &\qquad \leq \frac{n}{\log(n)}\cdot\log(r) + \left(1 + \frac{\varepsilon}{2}\right)\log(n) + 7 \\
    &\qquad \leq \left(1 + \frac{\varepsilon}{2}\right)\left(\frac{n}{\log(n)} + \log(n)\right) + 7,
\end{align*}
where we have used the fact that $\log(r) < 1 + \varepsilon/2$. Now, since $n/\log(n) = \omega(\log(n))$, there is some $N''(\varepsilon)$ such that for all $n \geq N''(\varepsilon)$, we have:
\begin{equation*}
    (1 + \varepsilon/2)\log(n) < (\varepsilon/2)\frac{n}{\log(n)}.
\end{equation*}
In particular, the number of rounds becomes further bounded as:
\begin{align*}
    \left(1 + \frac{\varepsilon}{2}\right)\left(\frac{n}{\log(n)} + \log(n)\right) + 7 &= \left(1 + \frac{\varepsilon}{2}\right)\frac{n}{\log(n)} + \left(1 + \frac{\varepsilon}{2}\right)\log(n) + 7 \\
    &< \left(1 + \frac{\varepsilon}{2}\right)\frac{n}{\log(n)} + \frac{\varepsilon}{2}\cdot\frac{n}{\log(n)} + 7 \\
    &= (1 + \varepsilon)\frac{n}{\log(n)} + 7.
\end{align*}
Hence, as long as $n \geq \max(N_r, N'(\varepsilon), N''(\varepsilon))$, all of the above conditions are satisfied, and the strategy is well-defined and finishes with $(1 + \varepsilon)\frac{n}{\log(n)} + 7$ rounds. Note that each of $N'(\varepsilon)$ and $N''(\varepsilon)$ depends only on $\varepsilon$, and $N(r)$ depends on $r$, which depends only on $\varepsilon$ as well. Therefore, for $n < \max(N_r, N'(\varepsilon), N''(\varepsilon))$, the number of quantifiers can be absorbed into the additive $O_\varepsilon(1)$ term.
\end{proof}

\anyvsanylower*

\begin{proof}
Take $n$ to be sufficiently large, and suppose $k$ (as a function of $n$) is the minimum number of quantifiers that is sufficient to separate every pair of disjoint sets of $n$-bit strings. We already know $k = o(n)$ from Theorem \ref{thm:all-vs-all upper bound}, and also $k \geq \log(n)$ from Proposition \ref{prop:one-vs-one-n}. Note that this means:
\begin{equation}\label{eqHelpful}
    \log(n) > \log(k) + 2/k.
\end{equation}
We wish to show that $k \geq n/\log(n)$. To this end, consider the number of pairwise inequivalent sentences that can be written with $k$ or fewer quantifiers. Assume any such sentence is written in prenex form. It must start with a quantifier prefix of length at most $k$, followed by its quantifier-free part, which can be written as a disjunction of types. The number of such quantifier prefixes is at most $\sum_{i = 0}^k2^i \leq 2^{k+1}$. Any type with $k$ or fewer variables can be completely specified by fixing the relative ordering of those variables (requiring at most $k$ occurrences of the variables, using transitivity of the $\leq$ relation), and fixing each of them to be $0$ or $1$ using the appropriate unary predicate (requiring another at most $k$ occurrences). Therefore, the total number of such types is at most $k!\cdot 2^k$. Since $k! \leq (k/2)^k$ for $k \geq 6$, the total number of types is bounded above by $(k/2)^k\cdot 2^k = 2^{k\log(k)}$. Any subset of types can be in the disjunction, leading to the number of quantifier-free parts being at most $2^{2^{k\log(k)}}$. This puts the total number of pairwise inequivalent formulas using $k$ quantifiers to be at most $2^k\cdot 2^{2^{k\log(k)}}$.
    
Now, consider an instance $(\cA, \{0, 1\}^n - \cA)$, where $\cA$ is a nonempty strict subset of the $n$-bit strings. Observe that any two distinct such instances \emph{must} require inequivalent sentences to separate them. Therefore, the number of pairwise inequivalent sentences we require in order to be assured of solving the problem is at least the number of such instances, which is $2^{2^n} - 2$, where we subtract $2$ to ensure there is at least one string on either side of each such instance. It follows that we need $2^k\cdot 2^{2^{k\log(k)}}  \geq 2^{2^n} - 2 \geq 2^{2^n - 1}$, i.e.:
\begin{equation}\label{k-need}
    k + 2^{k\log(k)} \geq 2^n - 1.
\end{equation}
But if $k < n/\log(n)$, we must have:
\begin{equation*}
    k + 2^{k\log(k)} < 2^{k\log(k) + 1} < 2^{k(\log(n) - 2/k) + 1} = 2^{k\log(n) - 1} < 2^{n - 1} < 2^n - 1,
\end{equation*}
where the first inequality follows because $2^{k\log(k)} > k$, the second follows from Eq.~\eqref{eqHelpful}, the third follows by the assumption that $k < n/\log(n)$, and the fourth follows for all sufficiently large $n$. Since this contradicts Eq.~\eqref{k-need}, it follows that $k \geq n/\log(n)$, as desired. In fact, the same argument also shows that with high probability, a \emph{random} instance $(\cA, \{0, 1\}^n - \cA)$ requires at least $n/\log(n)$ quantifiers to separate.
\end{proof}

\end{document}